\documentclass[review,1p,times,authoryear]{elsarticle}

\usepackage{amsmath}
\usepackage{amsthm}
\usepackage{amssymb}
\usepackage{amsfonts}

\usepackage{mathdots}

\usepackage{mathptmx}
 \usepackage{hyperref}

\usepackage[capitalize]{cleveref}

 \usepackage{mathtools}

 \usepackage{algorithm}
 \usepackage{algorithmic}

 \usepackage{url}
 \usepackage{enumitem}

 \usepackage{todonotes}
 
\usepackage{xspace}

\usepackage[makeroom]{cancel}

\renewcommand{\dim}{\mathtt{dim}}
\renewcommand{\ker}{\mathtt{Ker}}

\newcommand{\egcd}{\textsc{eGCD}\xspace}

\newcommand{\rk}{\mathtt{rk}}
\newcommand{\costMul}{\mathtt{M}}

\newcommand{\field}{\mathbb{K}}
\newcommand{\K}{\mathbb{K}}
\newcommand{\C}{\mathbb{C}}
\newcommand{\R}{\mathbb{R}}
\newcommand{\closurefield}{\mathbb{\overline{\field}}}
\newcommand{\transpose}{\mathsf{T}}

\newcommand{\N}{\mathbb{N}\xspace}

\newcommand{\Uch}{\mathbb{U}}

\newcommand{\ZZ}{\mathbb{Z}\xspace}
\newcommand{\QQ}{\mathbb{Q}\xspace}
\newcommand{\RR}{\mathbb{R}\xspace}
\newcommand{\CC}{\mathbb{C}\xspace}

\newcommand{\PP}{\mathbb{P}\xspace}

\newcommand{\PC}{\PP^1(\closurefield)\xspace}
\newcommand{\PK}{\PP^1(\K)\xspace}

\newcommand{\OB}{O_B\xspace}
\newcommand{\sOB}{\widetilde{O}_B\xspace}
\newcommand{\sO}{\widetilde{O}\xspace}

\let\*\cdot

\newtheorem{theorem}{Theorem}
\newtheorem{lemma}[theorem]{Lemma}
\newtheorem{corollary}[theorem]{Corollary}

\newtheorem{proposition}[theorem]{Proposition}
\newtheorem{propdef}[theorem]{Proposition-Definition}
\newtheorem{definition}[theorem]{Definition}

\newtheorem{remark}[theorem]{Remark}

\crefname{equation}{equation}{equations}
\crefname{section}{section}{sections}
\crefname{lemma}{lemma}{lemmata}
\crefname{proposition}{proposition}{propositions} 
\crefname{propdef}{proposition-Definition}{definitions} 
\crefname{remark}{remark}{remarks}
\crefname{enumi}{step}{steps}
\crefname{theorem}{theorem}{theorems}
\crefname{algorithm}{algorithm}{algorithms}
\crefname{figure}{figure}{figures}
\crefname{definition}{definition}{definitions}
\crefname{observation}{observation}{observations}
\crefname{corollary}{corollary}{corollaries}
\crefname{appendix}{appendix}{appendices}

\newcommand{\generators}[1]{\langle #1 \rangle}

\usepackage{subcaption}

\usepackage{caption}

 \newfont{\subsecit}{ptmbi8t at 11pt}  % 

\begin{document}

\begin{frontmatter}

  \title{A nearly optimal algorithm to decompose binary forms}

  \author{Mat\'{i}as R. Bender}
  \address{Sorbonne Universit\'e, \textsc{CNRS}, \textsc{INRIA},
    Laboratoire d'Informatique de Paris~6, \textsc{LIP6},
    \'Equipe \textsc{PolSys},
    4 place Jussieu, F-75005, Paris, France}
  \ead{matias.bender@inria.fr}
  \ead[http://www-polsys.lip6.fr/~bender/]{http://www-polsys.lip6.fr/~bender/}

  \author{Jean-Charles Faug\`ere}
  \address{Sorbonne Universit\'e, \textsc{CNRS}, \textsc{INRIA},
    Laboratoire d'Informatique de Paris~6, \textsc{LIP6},
    \'Equipe \textsc{PolSys},
    4 place Jussieu, F-75005, Paris, France}
  \ead{jean-charles.faugere@inria.fr}
  \ead[http://www-polsys.lip6.fr/~jcf/]{http://www-polsys.lip6.fr/~jcf/}

  \author{Ludovic Perret}
  \address{Sorbonne Universit\'e, \textsc{CNRS}, \textsc{INRIA},
    Laboratoire d'Informatique de Paris~6, \textsc{LIP6},
    \'Equipe \textsc{PolSys},
    4 place Jussieu, F-75005, Paris, France}
  \ead{ludovic.perret@lip6.fr}
  \ead[http://www-polsys.lip6.fr/~perret/]{http://www-polsys.lip6.fr/~perret/}
  
  \author{Elias Tsigaridas}
  \address{Sorbonne Universit\'e, \textsc{CNRS}, \textsc{INRIA},
    Laboratoire d'Informatique de Paris~6, \textsc{LIP6},
    \'Equipe \textsc{PolSys},
    4 place Jussieu, F-75005, Paris, France}
  \ead{elias.tsigaridas@inria.fr}
  \ead[http://www-polsys.lip6.fr/~elias/]{http://www-polsys.lip6.fr/~elias/}

  \begin{abstract}
    Symmetric tensor decomposition is an important problem with
    applications in several areas, for example signal processing,
    statistics, data analysis and computational neuroscience.
    It is equivalent to Waring's problem for homogeneous
    polynomials, that is to write a homogeneous polynomial in $n$
    variables of degree $D$ as a sum of $D$-th powers of linear forms,
    using the minimal number of summands.
    This minimal number is called the \emph{rank} of the
    polynomial/tensor.
    We focus on decomposing binary forms, a problem that corresponds
    to the decomposition of symmetric tensors of dimension $2$ and
    order $D$, that is, symmetric tensors of order $D$ over the vector
    space $\K^2$.
    Under this formulation, the problem finds its roots in invariant
    theory where the decompositions are related to canonical forms.
    %%
    %In recent years, those algorithms were extended for the general
    %symmetric tensor decomposition problem.
    
    We introduce a \emph{superfast} algorithm that exploits results
    from \emph{structured linear algebra}. It achieves a \emph{softly
      linear} arithmetic complexity bound. To the best of our
    knowledge, the previously known algorithms have at least quadratic
    complexity bounds.
    Our algorithm computes a symbolic decomposition in
    $O(\costMul(D) \log(D))$ arithmetic operations, where
    $\costMul(D)$ is the complexity of multiplying two polynomials of
    degree $D$.
    It is deterministic when the decomposition is unique. When the
    decomposition is not unique, it is randomized. We also present a
    Monte Carlo variant as well as a modification to make it a Las
    Vegas one.

    From the symbolic decomposition, we approximate the terms of the
    decomposition with an error of $2^{-\varepsilon}$, in
    $O\big(D \log^2(D)\big(\log^2(D) + \log(\varepsilon)\big)\big)$
    arithmetic operations. We use results from
    \citet{kaltofen1989improved} to bound the size of the
    representation of the coefficients involved in the decomposition
    and we bound the algebraic degree of the problem by
    $\min(rank, D-rank+1)$. We show that this bound can be tight.
    When the input polynomial has integer coefficients, our algorithm
    performs, up to poly-logarithmic factors,
    $\sOB(D \ell + D^4 + D^3 \tau)$ bit operations, where $\tau$ is
    the maximum bitsize of the coefficients and $2^{-\ell}$ is the
    relative error of the terms in the decomposition.
  \end{abstract}

  
\end{frontmatter}

\section{Introduction}
\label{sec:intro}

The problem of decomposing a symmetric tensor consists in writing it
as the sum of rank-$1$ symmetric tensors, using the minimal number of
summands. This minimal number is known as the rank of the symmetric
tensor\footnote{Some authors, e.g.,~\cite{comon2008symmetric}, refer to
  this number as the symmetric rank of the tensor.}.  The symmetric
tensors of rank-$1$ correspond to, roughly speaking, the $D$-th
outer-product of a vector. The decomposition of symmetric tensor is a
common problem which appears in divers areas such as signal
processing, statistics, data mining, computational neuroscience,
computer vision, psychometrics, chemometrics, among others. For a
modern introduction to the theory of tensor, their
decompositions and applications we refer to
e.g.,~\cite{comon2014tensors,landsberg2012tensors}.

There is an equivalence between decomposing symmetric tensors and
solving Waring's problem for homogeneous polynomials,
e.g.,~\cite{comon2008symmetric,helmke1992waring}. Given a symmetric
tensor of dimension $n$ and order $D$, that is a symmetric tensor of
order $D$ over the vector space $\K^n$, we can construct a homogeneous
polynomial in $n$ variables of degree $D$.
We can identify the symmetric tensors of rank-$1$ with the $D$-th
power of linear forms. 
Hence, to decompose a symmetric tensor of order $D$ is equivalent to
write the corresponding polynomial as a sum of $D$-th powers of linear
forms using the minimal numbers of summands. This minimal number is
the rank of the polynomial/tensor.

Under this formulation, symmetric tensor decomposition dates back to
the origin of modern (linear) algebra as a part of Invariant
Theory. In this setting, the decomposition of generic symmetric
tensors corresponds to canonical forms
\citep{sylvester1851canonicalforms,sylvester1851remarkablediscovery,
  gundelfinger1887theorie}. Together with the theory of apolarity,
this problem was of great importance because the decompositions
provide information about the behavior of the polynomials under linear
change of variables \citep{kung1984invariant}.

\paragraph{Binary Form Decomposition} We study the
decomposition of symmetric tensors of order $D$ and dimension $2$.
In terms of homogeneous polynomials, we consider a binary form
\begin{equation}
  \label{eq:bf}
  f(x, y) := \sum\nolimits_{i = 0}^D \tbinom{D}{i} a_i x^i y^{D-i},
\end{equation}
where $a_i \in \field \subset \CC$ and $\field$ is some field of
characteristic zero. We want to compute a decomposition
\begin{equation}
  \label{eq:decompnoLambda}
  f(x, y) = \sum\nolimits_{j = 1}^r (\alpha_j x + \beta_j y)^D ,
\end{equation}
where
$\alpha_1, \dots, \alpha_r, \beta_1, \dots, \beta_r \in
\closurefield$, with $\closurefield$ being the algebraic closure of
$\field$, and $r$ is minimal.
We say that a decomposition \textit{unique} if, for all the
decompositions, the set of points
$\{(\alpha_j, \beta_j) : 1 \leq j \leq r \}\subset
\PC$ is unique, where
$\PC$ is the projective space of
$\closurefield$ \citep{reznick2013length}.

\paragraph{Previous work} The decomposition of binary forms,
\Cref{eq:decompnoLambda}, has been studied extensively for
$\field = \mathbb{C}$. More than one century ago
\cite{sylvester1851remarkablediscovery,sylvester1851canonicalforms}
described the necessary and sufficient conditions for a decomposition
to exist, see \Cref{sec:sylvester_s_theorem}. He related the
decompositions to the kernel of Hankel matrices.  For a modern
approach of this topic, we refer to
\cite{kung1984invariant,kung1990canonical,
  reznick2013length,iarrobino1999power}. Sylvester's work was extended
to a more general kind of polynomial decompositions that we do not
consider in this work, e.g.,
\cite{gundelfinger1887theorie,reznick1996homogeneous,iarrobino1999power}.

Sylvester's results lead to an
algorithm (\Cref{alg:common}) to decompose binary forms
\citep[see][Sec.~3.4.3]{comon1996decomposition}.  In the case where the binary
form is of odd degree, then we can compute the decompositions using
Berlekamp-Massey algorithm
% \cite{berlekamp1966nonbinary,massey1969shift}
\cite[see][]{dur1989}. When the decomposition is unique, the
Catalecticant algorithm, which also works for symmetric tensors of
bigger dimension \citep{iarrobino1999power,oeding2013eigenvectors},
improves Sylvester's work. For an arbitrary binary form,
\cite{helmke1992waring} presented a randomized algorithm based on
Pad\'e approximants and continued fractions, in which he also
characterized the different possible decompositions. Unfortunately,
all these algorithms have complexity at least quadratic in the degree
of the binary form.

Besides the problem of computing the decomposition(s) many authors
considered the subproblems of computing the rank and deciding whether
there exists a unique decomposition, e.g.,
\cite{sylvester1851remarkablediscovery,sylvester1851canonicalforms,
  helmke1992waring,comas2011rank,bernardi2011computing}.  For example,
\cite{sylvester1851remarkablediscovery,sylvester1851canonicalforms}
considered generic binary forms, that is binary forms with
coefficients belonging to a dense algebraic open subset of
$\closurefield^{D+1}$ \cite[Section~3]{comon1996decomposition}, and
proved that when the degree is $2k$ or $2k+1$, for $k \in \N$, the
rank is $k+1$ and that the minimal decomposition is unique only when
the degree is odd. In the non-generic case,
\cite{helmke1992waring,comas2011rank,iarrobino1999power}, among
others, proved that the rank is related to the kernel of a Hankel
matrix and that the decomposition of a binary form of degree $2k$ or
$2k-1$ and rank $r$, is unique if and only if $r \leq k$.  With
respect to the problem of computing the rank there are different
variants of algorithms, e.g.,
\cite{comas2011rank,comon2008symmetric,bernardi2011computing}.  Even
though there are not explicit complexity estimates, by exploiting
recent superfast algorithms for Hankel matrices
\citep{pan2001structured}, we can deduce a nearly-optimal arithmetic
complexity bound for computing the rank using the approach of
\cite{comas2011rank}.

For the general problem of symmetric tensor decomposition, Sylvester's
work was successfully extended to cases in which the decomposition is
unique, e.g., \cite{brachat2010symmetric,oeding2013eigenvectors}.
There are also homotopy techniques to solve the general problem, e.g.,
to decompose generic symmetric tensors \citep{bernardi_tensor_2017} or,
when there is a finite number of possible decompositions and we know
at least one of them, to compute all the other decompositions
\citep{hauenstein_homotopy_2016}.
% We do not compare our algorithm
% against homotopy methods to decompose binary forms because in our case
% the decomposition are not generic and, when there is a finite number
% of decompositions, the decomposition is unique.
There are no complexity estimations for these methods.
Besides tensor decomposition, there are other related decompositions
for binary forms and univariate polynomials that we do not consider in
this work, e.g.,
\cite{reznick1996homogeneous,reznick2013some,giesbrecht2003algorithms,giesbrecht2010interpolation,Garcia-Marco:2017:RAS:3087604.3087605}.

\paragraph{Formulation of the problem}
Instead of decomposing the binary form as in \Cref{eq:decompnoLambda},
we compute $\lambda_1 \dots\lambda_r$, $\alpha_1 \dots \alpha_r$,
$\beta_1 \dots \beta_r \in \closurefield$, where $r$ is minimal,
such that,
\begin{align}
  \label{eq:bfDecomp}
  f(x, y) = \sum\nolimits_{j = 1}^r \lambda_j (\alpha_j x + \beta_j y)^D.
\end{align}
\noindent Since every $\lambda_j$ belongs to the algebraic closure of
the field $\K$, the problems are equivalent. This approach allows us to
control the algebraic degree
\citep{bajaj1988algebraic,nie2010algebraic} of the parameters
$\lambda_j$, $\alpha_j$, and $\beta_j$ in the decomposition
(Section~\ref{sec:algebraicDegree}).

Note that if the field is not algebraically closed and we force the
parameters to belong to the base field, that is
$\lambda_j, \alpha_j, \beta_j \in \field$, the decompositions induced
by \Cref{eq:decompnoLambda} and \Cref{eq:bfDecomp} are not equivalent.
We do not consider the latter case and we refer to
\cite{helmke1992waring,reznick1992sums,comon2008symmetric,boij2011monomials,blekherman2015typical}
for $\field = \RR$, and to
\cite{reznick1996homogeneous,reznick2013length,reznick2017binary}
for 
$\field \subset \CC$.

\paragraph{Main results}
We extend Sylvester's algorithm to achieve a nearly-optimal
complexity bound in the degree of the binary form.
By considering structural properties of the Hankel matrices, we
restrict the possible values for the rank of the decompositions and
we identify when the decomposition is unique.
We build upon \cite{helmke1992waring} and we use the Extended
Euclidean Algorithm to deduce a better complexity estimate  than what
was previously known.
Similarly to Sylvester's algorithm, our algorithm decomposes
successfully any binary form, without making any assumptions on the
input.

First, we focus on {\em symbolic decompositions}, that is a
representation of the decomposition as a sum of a rational function
evaluated at the roots of a univariate polynomial
(Definition~\ref{def:symbolicDecomp}).
We introduce an algorithm to compute a symbolic decomposition of a
binary form of degree $D$ in $O(\costMul(D) \log(D))$, where
$\costMul(D)$ is the arithmetic complexity of polynomial
multiplication (\Cref{thm:complexityAlgorithm}).  When the
decomposition is unique, the algorithm is deterministic and this is a
worst case bound.
When the decomposition is not unique, our algorithm makes some
random choices to fulfill certain genericity assumptions; thus the
algorithm is a Monte Carlo one.  However, we can verify if the
genericity assumptions hold within the same complexity bound, that
is $O(\costMul(D) \log(D))$, and hence we can also deduce a Las
Vegas variant of the  algorithm.

Following the standard terminology used in structured matrices
\citep{pan2001structured}, our algorithm is \textit{superfast} as its
arithmetic complexity matches the size of the input up to
poly-logarithmic factors.
The symbolic decomposition allow us to approximate the terms in a
decomposition, with a relative error of $2^{-\varepsilon}$, in
$O\big(D \log^2(D)\big(\log^2(D) + \log(\epsilon)\big)\big)$
arithmetic operations \citep{PanOpt02,numMethodsRoots}.
Moreover, we can deduce for free the rank and the border rank of the
tensor, see \cite[Section~1]{comas2011rank}.

Using results from \citet{kaltofen1989improved}, we bound the
algebraic degree of the decompositions by $\min(\text{rank}, D-\text{rank}+1)$
(\Cref{cor:dg-irr-Q}). Moreover, we prove lower bounds for the
algebraic degree of the decomposition and we show that in certain
cases the bound is tight (\Cref{sec:lower-bounds-algebr}). For
polynomials with integer coefficients, we bound the bit complexity, up
to poly-logarithmic factors, by $\sOB( D \ell + D^4 + D^3 \tau)$,
where $\tau$ is the maximum bitsize of the coefficients of the input
binary form and $2^{-\ell}$ is the error of the terms in the
decomposition (\Cref{thm:bf-decomp-bit-compl}).  This Boolean worst
case bound holds independently of whether the decomposition is unique
or not.

This work is an extension of the conference paper
\citep{bender2016superfast}.
With respect to the conference version, our main algorithm
(\Cref{alg:algorithmDecomp}) omits an initial linear change of
coordinates as we now rely on fewer genericity assumptions. In
contrast with our previous algorithm, we present an algorithm which is
deterministic when the decomposition is unique
(\Cref{thm:complexityAlgorithm}). When the decomposition is not
unique, our algorithm is still randomized but we present bounds for
the number of bad choices that it could make
(\Cref{thm:badValuesToInterpolateQ}).
With respect to the algebraic degree of the problem, we study the
tightness of the bounds that we proposed in the conference paper
(\Cref{theo:boundOfAlgDegree}). We introduce explicit lower bounds
showing that our bounds can be tight
(\Cref{sec:lower-bounds-algebr}).

\paragraph*{Organization of the paper}

First, we introduce the notation. In \Cref{sec:preliminaries}, we
present the preliminaries that we need for introducing our
algorithm. We present Sylvester's algorithm
(\Cref{sec:sylvester_s_theorem}), we recall some properties of the
structure of the kernel of the Hankel matrices
(\Cref{sec:kernelHankel}), we analyze its relation to rational
reconstructions of series/polynomials
(\Cref{sec:rational-reconstruction}), and we present the Extended
Euclidean Algorithm (\Cref{sec:great-comm-divis}).
Later, in \Cref{sec:algorithm}, we present our main algorithm to
decompose binary forms (\Cref{alg:algorithmDecomp}) and its proof of
correctness (\Cref{sub:correctness_of_algorithm_decomposition}).
This algorithm uses \Cref{alg:gettingVandW} to compute the kernel of
a family of Hankel matrices, which we consider in
\Cref{sub:computingVandW}.
Finally, in \Cref{sec:complexity}, we study the algebraic degree of
the problem (\Cref{sec:algebraicDegree}), we present tight bounds
for it (\Cref{sec:lower-bounds-algebr}), and we analyze the
arithmetic (\Cref{sec:arithmetic_complexity}) and bit complexity of
\Cref{alg:algorithmDecomp} (\Cref{sec:bit}).

\paragraph*{Notation}    
We denote by $O$, respectively $\OB$, the arithmetic, respectively bit, complexity
and we use $\sO$, respectively $\sOB$, to ignore (poly-)logarithmic
factors. $\costMul(n)$ is the arithmetic complexity of
multiplying two polynomial of degree $n$.
Let $\field$ be a zero characteristic subfield of $\CC$, and
$\closurefield$ its algebraic closure.
If $v = (v_0,\dots,v_n)^\transpose$ then $P_v = P_{(v_0,\dots,v_n)} :=
\sum_{i=0}^n v_i x^i y^{n-i}$.
Given a binary form $f(x,y)$, we denote by $f(x)$ the univariate
polynomial $f(x) := f(x,1)$. By $f'(x)$ we denote the derivative of $f(x)$
with respect to $x$.
For a matrix $M$, $\rk(M)$ is its rank and $\ker(M)$ its kernel.

\section{Preliminaries} 
\label{sec:preliminaries} 

  \subsection{An algorithm based on Sylvester's theorem}
  \label{sec:sylvester_s_theorem}

  Sylvester's theorem (\Cref{theo:Silv1851}) relates the minimal
  decomposition of a binary form to the kernel of a Hankel matrix.
  Moreover, it implies an (incremental) algorithm for computing the
  minimal decomposition. The version that we present in
  \Cref{alg:common} comes from  \citet[Section~3.2]{comon1996decomposition}.
    
    \begin{definition}
      \label{def:hankelFamily}
      
      Given a vector $a = (a_0, \dots, a_D)^\transpose$, we denote by
       $\{H_a^k\}_{1 \leq k \leq D}$ the family of Hankel
      matrices indexed by $k$, where
      $H_a^k \in {\field}^{(D-k+1) \times (k+1)}$ and
      \begin{align}
        \label{eq:defHankel}
        H_a^k := \begin{pmatrix}
          a_0    & a_1       & \cdots & a_{k-1} & a_{k}   \\
          a_1    & a_2       & \cdots & a_{k}   & a_{k+1} \\
          \vdots  & \vdots    & \ddots & \vdots  & \vdots  \\
          a_{D-k-1} & a_{D-k}   & \cdots & a_{D-2} & a_{D-1} \\
          a_{D-k}   & a_{D-k+1} & \cdots & a_{D-1} & a_{D}
        \end{pmatrix} .
     \end{align}
     \end{definition}

     \noindent
     We may omit the index $a$ in $H_a^k$ when it is clear from the
     context.

    \begin{theorem}[Sylvester, 1851]
    \label{theo:Silv1851}
      Let $f(x,y) = \sum_{i = 0}^D \binom{D}{i} a_i x^i y^{D-i}$ with
      $a_i \in \field \subseteq \CC$. Also, consider a non-zero vector 
      $c = (c_0, \dots, c_r)^\transpose \in \field^{r+1}$, such
      that the polynomial      
      $$P_{c} = \sum\nolimits_{i=0}^{r} c_i \, x^i \, y^{r-i}
      = \prod\nolimits_{j = 1}^r (\beta_j x - \alpha_j y)$$ is
      square-free and $\alpha_j, \beta_j \in \closurefield$, for all
      $1 \leq j \leq r$.
      Then, there are $\lambda_1, \dots \lambda_r \in \closurefield$
      such that we can write $f(x, y)$ as 
      $$f(x,y)=\sum_{j=1}^r \lambda_j (\alpha_j x + \beta_j y)^D,$$
      if and only if
      $(c_0, \dots, c_r)^\transpose \in \ker(H^r_{a})$.
    \end{theorem}

    For a proof of \Cref{theo:Silv1851} we refer to
    \citet[Theorem~2.1~\&~Corollary~2.2]{reznick2013length}.
    \begin{algorithm}[h]      
    \caption{\xspace{ \textsc{IncrDecomp} \cite[Figure~1]{comon1996decomposition}}} 
    \label{alg:common}
      \begin{enumerate}
          \item $r := 1$
          \item Get a random $c \in \ker(H^r)$
          \item If $P_c$ is not square-free, $r := r+1$ and GO TO $2$
          \item Write $P_c$ as $\prod_{j = 1}^r (\beta_j x - \alpha_j y)$
          \item Solve the transposed Vandermonde  system:
            \begin{align}
              {
              \label{eq:gettingLambdas}
              \begin{pmatrix}
                \beta_1^D & \cdots & \beta_r^D   \\
                \beta_1^{D-1}\alpha_1 & \cdots & \beta_r^{D-1}\alpha_r   \\
                \vdots & \ddots & \vdots  \\
                \alpha_1^D & \cdots & \alpha_r^D
              \end{pmatrix}
                                      \begin{pmatrix}
                                        \lambda_1 \\ \vdots \\ \lambda_r
                                      \end{pmatrix}
              = 
              \begin{pmatrix}
                a_0 \\ \vdots \\ a_D
              \end{pmatrix}
              }.
            \end{align}
            
          \item Return $\sum_{j=1}^{r} \lambda_j (\alpha_j x + \beta_j y)^D$
      \end{enumerate}
    \end{algorithm}
    \Cref{theo:Silv1851} implies \Cref{alg:common}. This algorithm
    will execute steps 2 and 3 as many times as the rank.
    At the $i$-th iteration it 
    computes the kernel of $H^i$. The dimension of this kernel is $\leq
    i$ and each vector in the kernel has $i+1$ coordinates. As the
    rank of the binary form can be as big as the degree of the binary
    form, a straightforward bound for the arithmetic complexity of
    \Cref{alg:common} is at least cubic in the degree.

    We can improve the complexity of \Cref{alg:common} by a factor of
    $D$ by noticing that the rank of the binary form is either
    $\rk(H^{\left\lceil\frac{D}{2}\right\rceil})$ or $D -
    \rk(H^{\left\lceil\frac{D}{2}\right\rceil}) + 2$
    \cite[Section~3]{comas2011rank} \linebreak \cite[Theorem~B]{helmke1992waring}.
    Another way to compute the rank is by using minors
    \cite[Algorithm~2]{bernardi2011computing}.

    The bottleneck of the previous approaches is that they have to
    compute the kernel of a Hankel matrix. However, even if we know
    that the rank of the binary form is $r$,  the dimension of the
    kernel of $H^{r}$ can still be as big as $O(D)$; the same bound holds
    for the length of the vectors in the kernel.  Hence, the
    complexity is lower bounded by $O(D^2)$.

    Our approach avoids the incremental construction. We exploit the
    structure of the kernel of the Hankel matrices and we prove that
    the rank has only two possible values
    (\Cref{lem:rankOfDecomposition}), see also
    \citep[Section~3]{comas2011rank},
    \citep[Theorem~B]{helmke1992waring}, or
    \citep{bernardi2011computing}. Moreover, we use a compact
    representation of the vectors in the kernel. We describe them as a
    combination of two polynomials of degree $O(D)$.
 
  \subsection{Kernel of the Hankel matrices}
  \label{sec:kernelHankel}

  The Hankel matrices are among the most studied structured matrices
  \citep{pan2001structured}. They are related to polynomial
  multiplication. We present results about the structure of their
  kernel. For details, we refer to
  \citet[Chapter~5]{heinig1984algebraic}.

  \begin{proposition}
    Matrix-vector multiplication of Hankel matrices is equivalent to
    polynomial multiplication.
    Given two binary forms $A := \sum_{i = 0}^D a_i x^i y^{D-i}$ and
    $U := \sum_{i = 0}^k u_i x^i y^{k-i}$, consider
    $R := \sum_{i = 0}^{D+k} r_i x^i y^{D + k-i} = A \cdot U$. If we
    choose the monomial basis $\{y^{D+k},\dots,x^{D+k}\}$, then the
    equality $A \cdot U = R$ is equivalent to
    \Cref{eq:polMultiplication}, where the central submatrix of the
    left matrix is $H^{k}_{(a_0,\dots,a_D)}$
    (\Cref{def:hankelFamily}).
  \begin{align}\label{eq:polMultiplication}
    \begin{pmatrix}
      &           &         &            & a_{0}       \\
      &           &         & a_0        & a_{1}       \\
      &           & \iddots & \iddots    & \vdots      \\
      & a_0       & \cdots & a_{k-2} & a_{k-1} \\ \hline
      a_0    & a_1       & \cdots & a_{k-1} & a_{k}   \\
      a_1    & a_2       & \cdots & a_{k}   & a_{k+1} \\
      \vdots & \vdots    & \iddots & \vdots  & \vdots  \\
      a_{D-k} & a_{D-k+1} & \cdots & a_{D-1} & a_{D} \\ \hline
      a_{D-k+1} & a_{D-k+2} & \cdots  & a_{D} &         \\
      \vdots  & \iddots  & \iddots &     &       \\
      a_{D-1}  &  a_D         &         &         &        \\
      a_D     &           &         &            &       \\
    \end{pmatrix} 
    \begin{pmatrix}
      u_k \\
      \vdots \\
      u_1 \\
      u_0 \\
    \end{pmatrix}
    =
    \begin{pmatrix}
      r_{0} \\
      r_{1} \\
      \vdots \\
      r_{k - 1} \\ \hline
      r_k \\
      r_{k+1} \\
      \vdots \\
      r_D \\ \hline
      r_{D+1} \\ 
      \vdots \\
      r_{D + k -1} \\
      r_{D + k }
    \end{pmatrix} .
  \end{align} 
\end{proposition}
  
Consider a family of Hankel matrices $\{H^k_a\}_{1 \leq k \leq D}$ as
in \Cref{def:hankelFamily}. There is a formula for the dimension of
the kernel of each matrix in the family $\{H^k_a\}_{1 \leq k \leq D}$
that involves two numbers, $N_1^a$ and $N_2^a$. To be more specific,
the following holds:
    
    \begin{proposition}
    \label{prop:existenceN1andN2}
    For any family of Hankel matrices $\{H^k_a\}_{1 \leq k \leq D}$
    there are two constants, $N_1^a$ and $N_2^a$, such that the following hold:
        \begin{enumerate}
            \item $0 \leq N_1^a \leq N_2^a \leq D$.
            \item For all $k$, $1 \leq k \leq D$, it holds
              $\dim(\ker(H^k_a)) = \max(0; k-N_1^a) + \max(0; k-N_2^a)$.
            \item $N_1^a + N_2^a = D$.
        \end{enumerate}
    \end{proposition}

    We may omit the index $a$ in $N_{1}^{a}$ and  $N_2^{a}$ when it
    is clear from the context.

    \begin{figure}[ht]
    \centering
    \begin{subfigure}{0.47\textwidth}
      \centering
      \includegraphics[width=175px]{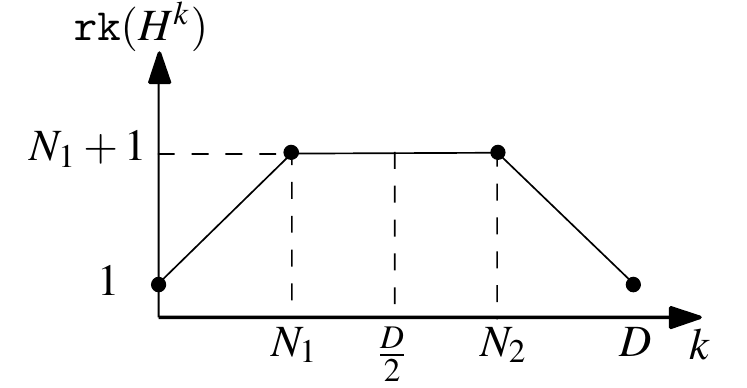}
      \caption{Rank of the Hankel matrices}
    \label{fig:relationHkandNs:rank}
    \end{subfigure}
    \begin{subfigure}{0.45\textwidth}
      \centering
      \includegraphics[width=175px]{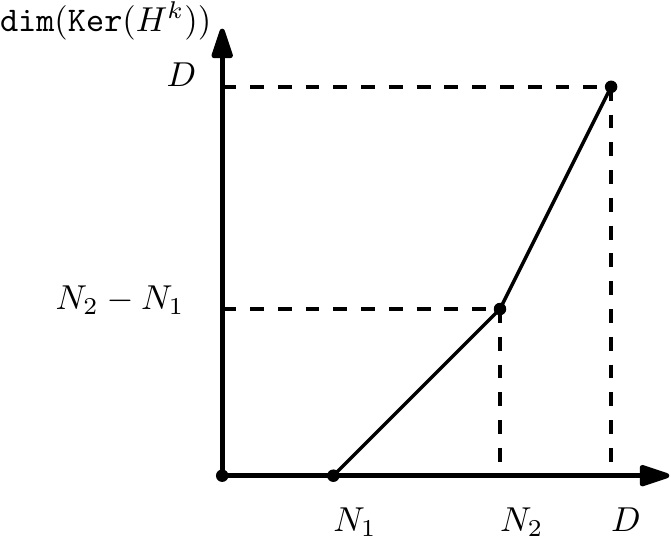}
      \caption{Dimension of the kernel}
    \label{fig:relationHkandNs:kernel}
    \end{subfigure}
    \caption{\normalsize Relation between $N_1$, $N_2$, and $D$}
    \label{fig:relationHkandNs}
  \end{figure}

    % \begin{figure}[h]
    %   \centering
    %   \includegraphics[]{dimKerHk.eps}
    %   \captionsetup{width=1\linewidth}
    %   \captionof{figure}{Dimension of the kernel of $H^k$}
    %   \label{fig:relationHkandNs}
    % \end{figure}
    
  An illustration of \Cref{prop:existenceN1andN2} appears in
  Figure~\ref{fig:relationHkandNs}.  The dimension of the kernel and
  the rank of the matrices are piece-wise-linear functions in $k$, the
  number of columns in the matrix.  The graphs of the functions
  consist of three line segments, as we can see in the Figures
  \ref{fig:relationHkandNs:rank} and \ref{fig:relationHkandNs:kernel}.
  The dimension of the kernel is an increasing function of $k$.  For
  $k$ from $1$ to $N_1$, the kernel of the matrix is trivial, so the
  rank increases with the number of columns. That is, the slope of the
  graph of the rank (Figure~\ref{fig:relationHkandNs:rank}) is $1$,
  while the slope of the graph of the dimension of the kernel
  (Figure~\ref{fig:relationHkandNs:kernel}) is $0$. For $k$ from
  $N_1+1$ to $N_2$, the rank remains constant as for each column that
  we add, the dimension of the kernel increases by one. Hence, the
  slope of the graph of the rank is $0$ and the slope of the graph of
  the dimension of the kernel is $1$. For $k$ from $N_2+1$ to $D$, the
  rank decreases because the dimension of the kernel increases by $2$,
  and so the slope of the graph of the rank is $-1$, while the slope
  of the graph of the dimension of the kernel is $2$.
  
  If $N_1 = N_2$, then the graph of both functions degenerates to two
  line segments. Regarding the graph of the rank, the first segment
  has slope $1$ for $k$ from 1 to $N_1+1$ and the second segment has
  slope $-1$ for $k$ from $N_1+1$ to $D$. For the graph of the
  dimension of the kernel, the first segment has slope $0$ from $1$ to
  $N_1+1$, and the second one has slope $2$ from $N_1$ to $D$.

  The elements of the kernel of the matrices in $\{H^k\}$ are
  related. To express this relation from a linear algebra point of
  view, we need to introduce the \textbf{U-chains}.
    
    \begin{definition}[\xspace{\citet[Definition~5.1]{heinig1984algebraic}}]
        A \textbf{U-chain} of length $k$ of a vector 
        $v = (v_0, \dots, v_{n})^\transpose \in {\field}^{n+1}$ is a set of vectors
        $\{ \Uch_k^0 v, \Uch_k^1 v, \dots , \Uch_k^{k-1} v\}\ \subset {\field}^{n+k}$.
        The $i$-th element, for $0 \leq i \leq  k-1$, is
        \begin{align*}
            \Uch_k^i v = (\underbrace{0,\, \dots \,,0}_{i}, \overbrace{v_0 ,\,\dots \,, v_{n}}^{n+1}, \underbrace{0, \, \dots \, ,0}_{k - 1 - i}) ,
        \end{align*}
        where $\Uch_k^i$ is an $i$-shifting matrix
        of dimension  $(n + k) \times (n+1)$
        \cite[page~11]{heinig1984algebraic}.
    \end{definition}

    If $v$ is not zero, then all the elements in a U-chain of $v$ are
    linearly independent. The following theorem uses U-chains to
    relate the vectors in the kernels of a family of Hankel matrices.

    \begin{proposition}[{Vectors {v} and {w}}]
        \label{prop:defVandW}
        Given a family of Hankel matrices $\{H^k\}_{1 \leq k \leq
        D}$, let $N_1$ and $N_2$ be the constants of \Cref{prop:existenceN1andN2}. 
      There are two vectors, $v \in
        {\field}^{N_1+2}$ and $w \in {\field}^{N_2+2}$, such that
        \noindent
        \begin{itemize}[leftmargin=*]
        \item If $0 \leq k \leq N_1$, then 
          $\ker(H^{k}) = \{ 0\}$.
        
        \item If $N_1 < k \leq N_2$, then the U-chain of $v$ of length $(k-N_1)$
        is a basis of $\ker(H^{k})$, that is  
        $$\ker(H^{k}) = \generators{ \Uch_{k-N_1}^0 v, \dots, \Uch_{k-N_1}^{k-N_1-1} v }.$$

        \item If $N_2 < k \leq D$, then the U-chain of $v$ of length $k-N_1$
        together with the U-chain of $w$ of length $k-N_2$ is a basis
        of $\ker(H^{k})$, that is  
        $$\ker(H^{k}) = \generators{ \Uch_{k-N_1}^0 v, \dots, \Uch_{k-N_1}^{k-N_1-1} v, \Uch_{k-N_2}^0 w, \dots, \Uch_{k-N_2}^{k-N_2-1} w }.$$
        
        \end{itemize}                        
    \end{proposition} 

    The vectors $v$ and $w$ of \Cref{prop:defVandW} are not
    unique. The vector $v$ could be any vector in $\ker(H^{N_1+1})$.
    The vector $w$ could be any vector in $\ker(H^{N_2+1})$ that does not
    belong to the vector space generated by the \mbox{U-chain} of $v$
    of length $N_2-N_1+1$. From now on, given a family of Hankel
    matrices, we refer to $v$ and $w$ as the vectors of
    \Cref{prop:defVandW}.

    Let $u$ be a vector in the kernel of $H^k$ and $P_u$ its
    corresponding polynomial (see Notation).
    We say that $P_u$ is a \textbf{kernel polynomial}.
    As  $ P_{\Uch^j_{k} v} = x^j y^{k-1-j} P_v $, we
    can write any kernel polynomial of a family of Hankel matrices as a
    combination of $P_v$ and $P_w$ \cite[Propositions~5.1 \&
    5.5]{heinig1984algebraic}. Moreover, $P_v$ and $P_w$
    are relatively prime.
    
    \begin{proposition}
        \label{prop:polInKernel}
        \label{prop:PvPwCoprimes}
        Consider any family of Hankel matrices
        $\{H^k\}_{1 \leq k \leq D}$. The  kernel polynomials
        $P_v$ and $P_w$ are relative prime. Moreover, for each $k$,
        the set of kernel polynomials of the matrix $H^k$
        is as follows:
        \begin{itemize}%[leftmargin=*] 
            \item If $ 0 < k \leq N_1$, then it is
                $\{0\}$. 
            \item If $N_1 < k \leq N_2$, then it is
                $\{P_{\mu}\*P_v : \mu \in \field^{k-N_1}\}$. 
            \item If $N_2 < k \leq D  $, then it is 
              $\{P_{\mu}\*P_v + P_{\rho}\*P_w : \mu \in \field^{k-N_1}, \rho \in \field^{k-N_2}\}$.
        \end{itemize}
    \end{proposition}    

    \begin{corollary} \label{thm:PwNotInPv}
      Let $\omega \in \ker(H^{N_2+1})$ such that
      $P_\omega \not\in \{P_\mu \* P_v : \mu \in \K^{N_2-N_1+1}\}$,
      then we can
      consider $\omega$ as the vector $w$ from \Cref{prop:defVandW}.    
    \end{corollary}

    \subsection{Rational Reconstructions}
    \label{sec:rational-reconstruction}
    \noindent
    A rational reconstruction for a series, respectively a polynomial,
    consists in approximating the series, respectively the polynomial,
    as the quotient of two polynomials.
    Rational reconstructions are the backbone of many problems e.g.,
    Pad\'e approximants, Cauchy Approximations, Linear Recurrent
    Sequences, Hermite Interpolation. They are related to Hankel
    matrices. For an introduction to rational reconstructions, we
    refer to \citet[Chapter~7]{bostan2017algorithmes} and references
    therein.
    
    \begin{definition}
      Consider $a := (a_0,\dots,a_D)^\transpose \in \K^{D+1}$
      and a polynomial $A := \sum_{i = 0}^D a_i x^i \in \K[x]$. Given
      a pair of univariate polynomials $(U,R)$, we say that they are a
      rational reconstruction of $A$ modulo $x^{D+1}$ if
      $ A \cdot U \equiv R \mod x^{D+1}$.
    \end{definition}

    % In particular, if we restrict our previous definition to the cases
    % when $\deg(U) + \deg(V) \leq D$ and $U(0) = 1$, then the rational
    % recontruction is called Pad\'e approximant.

    Such a reconstruction is not necessarily unique.  Our interest
    emanates from the relation between the rational reconstructions of
    $A$ modulo $x^{D+1}$ and the kernels of the family of Hankel
    matrices $\{H^k_a\}_k$.

    \begin{lemma} \label{thm:vectInKernelToRationalApprox}
      Consider $\omega \in \ker(H^k_a)$ and
      $(r_0,\dots,r_{k-1}) \in \K^k$ such that
      \begin{align*} 
        \begin{pmatrix}
         0 & \cdots      &            0 & a_{0}       \\
        \vdots & \iddots   & \iddots    & \vdots      \\ 
        0  & a_0       & \cdots     & a_{k-1} 
        \end{pmatrix} 
        \begin{pmatrix}
          \omega_0 \\
          \omega_1 \\
          \vdots \\
          \omega_k \\
        \end{pmatrix}
        =
        \begin{pmatrix}
          r_{0} \\
          \vdots \\
          r_{k-1}
        \end{pmatrix} .
      \end{align*}
      Then, $(P_\omega(1,x), \sum_{i = 0}^{k-1} r_i x^i)$ is a
      rational reconstruction of $A$ modulo $x^{D+1}$.
    \end{lemma}

    \begin{proof}
      %%
      % We can relate the rational reconstruction and the kernel
      % polynomials of Hankel matrix.
      Following
        \Cref{eq:polMultiplication}, if $\omega \in \ker(H^k_a)$, then
      \begin{align}  \label{eq:vectInKernToReconstr}
        \begin{pmatrix}
          &           &            & a_{0}       \\
          &           & \iddots    & \vdots      \\ 
          & a_0       & \cdots     & a_{k-1} \\ \hline
          a_0    & a_1       & \cdots  & a_{k}   \\
          \vdots & \vdots    & \iddots  & \vdots  \\
          a_{D-k} & a_{D-k+1} & \cdots & a_{D} \\
        \end{pmatrix} 
        \begin{pmatrix}
          \omega_0 \\
          \omega_1 \\
          \vdots \\
          \omega_k \\
        \end{pmatrix}
        =
        \begin{pmatrix}
          r_{0} \\
          \vdots \\
          r_{k-1} \\ \hline
          0 \\
          \vdots \\
          0
        \end{pmatrix} .
      \end{align}
      Hence, $P_\omega(1,x) = \sum_{i = 0}^k \omega_{k-i} x^i$ and
      $A \cdot P_\omega(1,x)  \equiv \sum_{i = 0}^{k-1} r_i x^i \mod
      x^{D+1}$. Therefore, \linebreak
      $(P_\omega(1,x), \sum_{i = 0}^{k-1} r_i x^i)$ is a
      rational reconstruction of $A$ modulo $x^{D+1}$.
    \end{proof}
  
    % The other way around of \cref{thm:vectInKernelToRationalApprox} is also plausible,
    % that is if  $A \, U = R \mod x^{D+1}$, then compute an element
    % in the kernel of $H^{\max(\deg(U),\deg(R)+1)}_a$.
 
    \begin{lemma} \label{thm:fromRationalToKernelPolynomial} If
      $(U,R)$ is a rational reconstruction of $A$ module $x^{D+1}$,
      then there is a vector \linebreak
      $\omega \in \ker(H^{\max(\deg(U),\deg(R)+1)}_a)$ such that
      $$P_\omega =
      U\!\!\left( \frac{y}{x} \right) \, x^{\max(\deg(U),\deg(R)+1)}.$$
    \end{lemma}

    \begin{proof}
      Let $k = \deg(U)$, $q = \deg(R)$, $U = \sum_i u_i x^i$ and
      $R = \sum_i r_i x^i$. Following
      \Cref{eq:polMultiplication}, $A \, U \equiv R \mod x^{D+1}$ is
      equivalent to,
      \begin{align} \label{eq:ofTheoremfromRationalToKernelPolynomial}
        \begin{pmatrix}
          &           &         & a_{0}       \\
          &           & \iddots & \vdots      \\
          & a_0       & \cdots  & a_{k-1} \\
          a_0    & a_1       & \cdots  & a_{k}   \\
          \vdots & \vdots    & \iddots & \vdots  \\
          a_{D-k} & a_{D-k+1} & \cdots  & a_{D} \\
        \end{pmatrix} 
        \begin{pmatrix}
          u_k \\
          \vdots \\
          u_1 \\
          u_0 \\
        \end{pmatrix}
        =
        \begin{pmatrix}
          r_{0} \\
          r_{1} \\
          \vdots \\
          r_{q} \\
          0 \\
          \vdots \\
          0
        \end{pmatrix} .
      \end{align}

      If $k > q$, \Cref{eq:ofTheoremfromRationalToKernelPolynomial}
      reduces to \Cref{eq:vectInKernToReconstr}, and so
      $\omega = (u_k, \dots, u_0) \in \ker(H^k_a)$. Moreover,
      $$U\!\left( \frac{y}{x} \right) x^k = \sum_{i = 0}^{k} u_i y^{i}
      x^{k-i} 
      \underset{(j \leftrightarrow k - i)}{=}
      \sum_{j = 0}^k u_{k - j}
      x^j y^{j - k} = P_\omega.$$

      If $q \geq k$, we extend the vector $(u_k,\dots,u_0)$ by adding
      $(q + 1 - k)$ leading zeros. We rewrite
      \Cref{eq:ofTheoremfromRationalToKernelPolynomial} as
      \Cref{eq:addingZerosToRationalReconstr}.
      The concatenation of the two bottom submatrices form the matrix
      $H^{q+1}_a$, and so
      $\omega = (0,\dots,0,u_k,\dots,u_0) \in \ker(H^{q+1}_a)$.  Also,
      $$P_\omega = \sum_{j = 0}^k u_j x^{q+1-j} y^{j} + \sum_{j =
        k+1}^{q+1} 0 \, x^{q+1-j} y^{j} = U \left( \frac{y}{x} \right)
      x^{q + 1}.$$
      % as,
      % \begin{align*}
      %   P_\omega = \sum_{i = 0}^{q+1} \omega_i x^i y^{q+1 - i} =
      %   \sum_{i = q+1-k}^{q+1} \omega_i x^i y^{q+1 - i} =
      %   \sum_{i = q+1-k}^{q+1} u_{q+1 - i} x^i y^{q+1 - i} \\
      %   \mathrel{\mathop{=}\limits_{(j \leftrightarrow q+1-i)}}
      %   \sum_{j = 0}^{k} u_{j} x^{q+1-j} y^{j} =
      %   x^{q+1} \sum_{j = 0}^{k} u_{j} \left( \frac{y}{x} \right)^{j} =
      %   U \left( \frac{y}{x} \right)
      %   x^{q+1}.
      % \end{align*}    
      %% 
      {\small
        \begin{align} \label{eq:addingZerosToRationalReconstr}
          \left[\begin{array}{ccc | cccc}
                  &      &      &             &           & a_{0}       \\
                  &      &      &             & \iddots    & \vdots      \\ % \hline
                  &      &      &      a_0    & \cdots    & a_{k}   \\ 
                  &      & a_0  &      a_1    & \cdots   & a_{k+1} \\
                  & \iddots &   \iddots    &  \vdots  & \iddots  & \vdots  \\ \hline
                  a_0  & \dots  &  a_{q - k}  &  a_{q + 1 -k}   & \iddots  & a_{q+1}  \\
                  \vdots & \iddots      &  \vdots  &      \vdots & \iddots  & \vdots  \\
                  a_{D-q-1}  & \dots & a_{D-k-1} &     a_{D-k}    & \cdots  & a_{D}
                \end{array}\right]
                                                                              \begin{pmatrix}
                                                                                0 \\ 
                                                                                \vdots \\ 
                                                                                0 \\  \hline
                                                                                u_{k} \\
                                                                                \vdots \\
                                                                                u_0 \\
                                                                              \end{pmatrix}
          =
          \begin{pmatrix}
            r_{0} \\
            \vdots \\ % \hline
            r_{k} \\ 
            \vdots \\
            r_{q} \\ \hline
            0 \\
            \vdots \\
            0
          \end{pmatrix} .
        \end{align}
      }
    \end{proof}

    \begin{remark} \label{thm:degreeKernelPol}
      If $(U,R)$ is a rational reconstruction, then the degree of the kernel polynomial
      $P_\omega(x,y) = U \!\! \left( \frac{y}{x} \right)
      x^{\max(\deg(U),\deg(R)+1)}$ is $\max(\deg(U), \deg(R) + 1)$.
      In particular, the maximum power of $x$ that divides the kernel
      polynomial $P_\omega$ is $x^{\max(0, \deg(R) + 1 - \deg(U))}$.
    \end{remark}

    % \begin{remark} \label{thm:fromKernelPolynomialToU}
    %   The relation
    %   $P_\omega(x,y) = U \left( \frac{y}{x} \right)
    %   x^{\max(\deg(U),\deg(R)+1)}$ in
    %   \Cref{thm:fromRationalToKernelPolynomial}, implies
    %   $P_\omega(1,x) = U(x)$.
    % \end{remark}

    \subsection{Greatest Common Divisor  and B\'ezout identity}
    \label{sec:great-comm-divis}
    
    The Extended Euclidean algorithm (\egcd) is a variant of the
    classical Euclidean algorithm that computes the Greatest Common
    Divisor of two univariate polynomials $A$ and $B$, $\gcd(A,B)$,
    together with two polynomials $U$ and $V$, called
    \textit{cofactors}, such that $U \, A + V \, B = \gcd(A,B)$. In the
    process of computing these cofactors, the algorithm computes a
    sequence of relations between $A$ and $B$ that are useful to solve
    various problems, in particular to compute the rational
    reconstruction of $A$ modulo $B$.
    For a detailed exposition of this algorithm, we refer to
    \citet[Chapter~6]{bostan2017algorithmes} and \citet[Chapter~3 and
    11]{gathen_modern_2013}.
       
    \begin{algorithm}[h] 
      \caption{Calculate the \egcd of $A$ and $B$} 
      \label{alg:egcd} 
      \begin{algorithmic}
        \STATE $(U_0, V_0, R_0) \leftarrow (0,1,B)$
        \STATE $(U_1, V_1, R_1) \leftarrow (1,0,A)$
        \STATE $k \leftarrow 1$

        \WHILE{$R_k \neq 0$}
        \STATE $k \leftarrow k + 1$
        \STATE $Q_{k - 1} \leftarrow R_{k - 2}  \text{ quo }  R_{k - 1}$
        \STATE $(U_k, V_k, R_k) \leftarrow (U_{k-2},V_{k-2},R_{k-2}) - Q_{k - 1} \, (U_{k-1},V_{k-1},R_{k-1})$
        % \STATE $r_k \leftarrow r_{k - 2} \text{ rem } r_{k - 1}$
        % \STATE $V_k \leftarrow V_{k - 2} - q_k * V_{k - 1}$
        % \STATE $v_k \leftarrow v_{k - 2} - q_k * v_{k - 1}$
        \ENDWHILE

        \STATE Return $\{(U_{i},V_{i},R_{i})\}_i$
      \end{algorithmic}  
    \end{algorithm}

    The Extended Euclidean Algorithm (\Cref{alg:egcd}) computes a
    sequence of triples $\{(U_i,V_i,R_i)\}_i$ which form the
    identities
    \begin{align} \label{eq:bezoutIdentity}
      U_i \, A + V_i \, B = R_i, \quad \text{for all }  i.
    \end{align}
    Following \citet{gathen_modern_2013}, we refer to these triplets as
    the rows of the Extended Euclidean algorithms of $A$ and
    $B$. Besides \Cref{eq:bezoutIdentity}, the rows are related to
    each other as follows.
    %: the degrees of $R_i$ form a strictly decreasing
    %sequence, $U_i$ and $V_i$ are coprime, and we can deduce the degree
    %of $U_i$ from the one of $R_{i-1}$.
    %%
    %\todo{review 2: Remarks / Lemmas 13--16 are redundant with the
    %  text preceding them. It probably makes sense to delete one
    %  instance of each phrase / sentence.} 
        
    \begin{remark}
      \label{thm:RiIsDecreasing}
      The degrees of the
      polynomials $\{R_i\}_i$ form a strictly decreasing sequence,
      that is $\deg(R_i) > \deg(R_{i+1})$ for every $i$.
    \end{remark}
    
    \begin{lemma}[{\citealp[Sec~7.1]{bostan2017algorithmes}}] \label{thm:bezoutFactorsCoprime}
      For each $i$, $U_{i} \, V_{i+1} - U_{i+1} \, V_i = (-1)^i$, and so
      the polynomials $U_i$ and $V_i$ are coprime.
    \end{lemma}

    % \begin{proof}
    %   By induction, for $i = 0$,
    %   $u_{0} \, v_{0+1} - u_{0+1} \, v_0 = 1*1 - 0*0 = (-1)^0$.
    %   If $u_{i-1} \, v_{i} - u_{i} \, v_{i-1} = (-1)^{i-1}$, then
    %   $u_{i} \, v_{i+1} - u_{i+1} \, v_i =
    %   u_{i} \, (v_{i-1} - q_{i} v_{i}) - (u_{i-1} - q_{i} u_{i}) \, v_i =
    %   u_{i} \, v_{i-1} - u_{i-1} \, v_{i} - q_i (u_i \, v_i - u_i \, v_i) =
    %   (-1) \cdot (u_{i-1} \, v_{i} - u_{i} \, v_{i-1}) = (-1) (-1)^{i-1} = (-1)^i$.
    % \end{proof}
    
    \begin{lemma}[{\citealp[Lem~7.1]{bostan2017algorithmes}}]
      \label{thm:relationOfDegreesEGCD}
      For each $i > 0$, the degree of $U_i$ is the degree of $B$
      minus the degree of $R_{i-1}$, that is
      $$\deg(U_{i}) = \deg(B) - \deg(R_{i-1}), \quad \forall\, i > 0.$$
    \end{lemma}

    \noindent Every row of the Extended Euclidean Algorithm leads to rational
    reconstruction of $A$ modulo $B$.

    \begin{remark} \label{thm:egcdGivesRatRecontrs}
      For each $i \geq 0$, $(U_i,R_i)$ is a rational reconstruction of
      $A$ modulo $B$.
    \end{remark}

\section{The Algorithm}
\label{sec:algorithm}

One of the drawbacks of \Cref{alg:common} and its variants is that
they rely on the computation of the kernels of many Hankel matrices
and they ignore the particular structure that it is present in all of
them. Using \Cref{lem:rankOfDecomposition}, we can skip many
calculations by computing only two vectors, $v$ and $w$
(\Cref{prop:defVandW}). This is the main idea behind
\Cref{alg:algorithmDecomp} that leads to a softly-linear arithmetic
complexity bound (\Cref{sec:arithmetic_complexity}).

\Cref{alg:algorithmDecomp} performs as follows: First,
\cref{stp:getVandW} computes the kernel polynomials $P_v$ and $P_w$
which, by \Cref{prop:polInKernel}, allow us to obtain the kernel
polynomials of all the Hankel matrices
(see~\Cref{sub:computingVandW}).
  Then, \cref{stp:choosingQ} computes a square-free kernel polynomial
  of the minimum degree $r$ (see~\Cref{sub:computing_a_square_free}).
  Next, \cref{stp:solveLambdas} computes the
  coefficients $\lambda_1, \dots, \lambda_r$
  (see~\Cref{sec:solving_the_lambdas}).
  Finally, \cref{stp:return} recovers a
  decomposition for the original binary form.

  Let $f$ be a binary form as in \Cref{eq:bf} and let
  $\{H^{k}\}_{1\leq k \leq D}$ be its corresponding family of Hankel
  matrices (see \Cref{def:hankelFamily}).
  The next well-known lemma establishes the rank of $f$.
  
  \begin{lemma}
    \label{lem:rankOfDecomposition}
    Assume $f$, $\{H^k\}_k$, $N_1$ and $N_2$ of
    \Cref{prop:existenceN1andN2}, and $v$ and $w$ of
    \Cref{prop:defVandW}. 
    If $P_v$ (\Cref{prop:polInKernel}) is square-free then the rank
    of $f$ is $N_1+1$, else, it is $N_2+1$.
  \end{lemma}
  
  \begin{proof}
    By \Cref{prop:existenceN1andN2}, for $k < N_1 + 1$, the kernel of
    $H^k$ is trivial. Hence, by Sylvester's theorem
    (\Cref{theo:Silv1851}), there is no decomposition with a rank
    smaller than $N_1+1$.
    Recall that $v \in \ker(H^{N_1+1})$.
    So, if $P_v$ is
    square-free, by Sylvester's theorem, there is a decomposition
    of rank $N_1+1$.
          
    Assume $P_v$ is not square-free.
    For $N_1+1 \leq k \leq N_2$,~$P_v$ divides all the kernel
    polynomials of the matrices $H^k$ (\Cref{prop:polInKernel}).
    Therefore, none of them is square-free, and so 
    the rank is at least $N_2 +1$.
    
    By \Cref{prop:PvPwCoprimes},
    $P_v$ and $P_w$ are coprime. 
    So, there is a polynomial
    $P_{\mu}$ of degree $N_2-N_1$ such that $Q_{\mu} := P_v\*P_{\mu} + P_w$ is
    square-free. A formal proof of this appears in
    \Cref{theo:nonZeroDiscriminant}. 
    By \Cref{prop:polInKernel},
    $Q_{\mu}$ is a square-free kernel polynomial of degree $N_2+1$.
    Consequently, by Sylvester's theorem, there is a decomposition
    with rank $N_2+1$.
  \end{proof}

  For alternative proofs of \Cref{lem:rankOfDecomposition} we refer to 
  \cite[Theorem~B]{helmke1992waring}, \cite[Section~3]{comas2011rank},
  \citep{bernardi2011computing}, or
  \cite[Section~4]{carlini_symmetric_2018}.
   
  To relate \Cref{lem:rankOfDecomposition} with the theory of binary
  form decomposition, we recall that the decompositions are identified
  with the square-free polynomials in the annihilator of the ideal
  $\langle f \rangle$ \citep{kung1984invariant};\cite[Chapter
  1]{iarrobino1999power}.  All the kernel polynomials of $\{H^k\}_k$
  belong to the annihilator of $\langle f \rangle$ and they form an
  ideal. If $f$ is a binary form of degree $D = 2k$ or $2k+1$, then
  this ideal is generated by two binary forms of degrees $\rk(H^{k})$
  and $D + 2 - \rk(H^{k})$, with no common zeros
  \cite[Theorem~1.44]{iarrobino1999power}. These are the polynomials
  $P_v$ and $P_w$.  Using this interpretation \Cref{alg:common} and
  its variants compute a (redundant) generating set of the annihilator
  of $\langle f \rangle$, while \Cref{alg:algorithmDecomp} computes a
  (minimal) basis.

  \begin{algorithm}[h]
    \caption{\textsc{FastDecomp}}
    \label{alg:algorithmDecomp}
    \begin{algorithmic}
      
      \REQUIRE A binary form $f(x, y)$ of degree $D$

      \ENSURE A decomposition for $f(x, y)$ of rank $r$.
      
      \begin{enumerate}[leftmargin=*]

      \item \label{stp:getVandW} \textbf{Compute {\subsecit $\mathbf{P_v}$} and
          {\subsecit $\mathbf{P_w}$} of $\mathbf{\{H^k_{\emph{a}}\}_k}$}

        We use \Cref{alg:gettingVandW} with  $(a_0,\dots,a_D)$.

        \item \label{stp:choosingQ} \textbf{\textbf{IF} $\mathbf{P_v(x,y)}$ is square-free},
          
          \qquad $Q \longleftarrow P_v$
          
          \qquad $r \longleftarrow N_1 + 1$  \COMMENT{The \emph{rank} of the decomposition is the degree of $Q$}
          
        \item[] \textbf{ELSE}
          
           \qquad \textbf{Compute a square-free binary form $\mathbf{Q}$ }
           
           \qquad We compute a vector $\mu$ of length $(N_2-N_1+1)$,

           \qquad such
           that $(P_\mu \* P_v + P_w)$ is square-free
           (\Cref{sec:getting_mu}).

           \qquad $Q \longleftarrow P_\mu \* P_v + P_w$

           \qquad  $r \longleftarrow N_2 + 1$ \COMMENT{The \emph{rank} of the decomposition is the degree of $Q$}
           
        \item \label{stp:solveLambdas} \textbf{Compute the coefficients}
          $\mathbf{\lambda_1,\dots,\lambda_r}$

          Solve the system of \Cref{eq:gettingLambdas} where
          $Q(x,y) = \prod_{j = 1}^r (\beta_j x - \alpha_j y)$.

        For details and the representation of $\lambda_j$, see \Cref{sec:solving_the_lambdas}.

        \item \label{stp:return} \textbf{Return}
        $ f(x, y) =  \sum_{j=1}^{r} \lambda_j ( \alpha_j x + \beta_j y )^D$
    \end{enumerate}
  \end{algorithmic}
  \end{algorithm} 

  \subsection[Computing the polynomials Pv and Pw]{Computing the polynomials
    {\subsecit $P_v$} and {\subsecit $P_w$}}
  \label{sub:computingVandW}

  We use \Cref{thm:vectInKernelToRationalApprox} and
  \Cref{thm:fromRationalToKernelPolynomial} to compute the polynomials
  $P_v$ and $P_w$ from \Cref{prop:polInKernel} as a rational
  reconstruction of $A := \sum_{i = 0}^D a_i x^i$ modulo $x^{D+1}$.
  Our algorithm exploits the Extended Euclidean Algorithm in a similar
  way as \cite{cabay1986algebraic} do to compute scaled Pad\'e
  fractions.
  
  In the following, let $v$ be the vector of \Cref{prop:defVandW},
  consider $U_v := P_v(1,x)$ and $R_v \in \K[x]$ as the remainder of
  the division of $(A \cdot P_v(1,x))$ by $x^{D+1}$. Note that the
  polynomial $R_v$ is the unique polynomial of degree strictly smaller
  to $N_1 + 1$ such that $A \cdot P_v(1,x) \equiv R_v \mod x^{D+1}$,
  see \Cref{eq:vectInKernToReconstr}.

  \begin{lemma} \label{thm:rationalReconstrOfLowDegree}
    If $(U,R)$ is a rational reconstruction of $A$ modulo $x^{D+1}$ such
    that $\max(\deg(U), \linebreak \deg(R) + 1) \leq N_2$, then there is a
    polynomial $Q \in \K[x]$ such that
    $Q \cdot P_v(x,1) = U$ and $Q \cdot R_v = R$.
  \end{lemma}

    \begin{proof}
      Let $k : =\deg(U)$ and $q := \deg(R)$. By
      \Cref{thm:fromRationalToKernelPolynomial}, there is a non-zero
      vector \linebreak $\omega \in \ker(H^{\max(k,q +1)}_a)$ such
      that the kernel polynomial $P_\omega$ is equal to
      $U\!\left( \frac{y}{x} \right) x^{\max(k,q+1)}$. Hence,
      $\ker(H^{\max(k,q +1)}_a) \neq 0$ and so, by
      \Cref{prop:defVandW}, $N_1 < \max(k,q +1)$.  We assume that
      $\max(k,q+1) \leq N_2$, hence the degree of $P_\omega$ is
      smaller or equal to $N_2$ and, by \Cref{prop:polInKernel},
      $P_\omega$ is divisible by $P_v$. Therefore, there is a
      polynomial $\bar{Q} \in \K[x,y]$ such that
      $\bar{Q} P_v = P_\omega$.
      Let $Q \,:= \bar{Q}(1,x)$. By definition, $U_v = P_v(1,x)$ and
      $U = P_\omega(1,x)$, so $U = Q \, U_v$.
      Hence, $Q \, R_v \equiv R \mod x^{D+1}$, because
%%
%%      We assume
%%      that definition , and so, by
%%      multiplying by $Q$, 
%%      $Q \,R_v \equiv R \mod x^{D+1}$
      $R_v \equiv U_v \, A \mod x^{D+1}$ and
      $Q \,R_v \equiv Q \, U_v \, A \equiv U A \equiv R \mod x^{D+1}$.
      If the degrees of $(Q \, R_v)$ and $R$ are smaller than $D + 1$,
      then $Q \,R_v = R$, as we want to prove. By assumption,
      $\deg(R) < N_2 \leq D$ and $\deg(U_v \, Q) = \deg(U) \leq N_2$.
      By definition, the degree of $R_v$ is smaller of equal to $N_1$,
      and so
      $\deg(Q \, R_v) \leq \deg(U_v \, Q \, R_v) \leq N_2 + N_1 = D$
      (\Cref{prop:existenceN1andN2}).
    \end{proof}

    We can use this lemma to recover the polynomial $P_v$ from certain
    rational reconstructions.
    
    \begin{corollary}
      If $(U,R)$ is a rational reconstructions of $A$ modulo $x^{D+1}$
      such that $\max(\deg(U), \linebreak \deg(R) + 1) \leq N_2$ and
      for every polynomial $Q$ of degree strictly bigger than zero
      that divides $U$ and $R$, $(\frac{U}{Q},\frac{R}{Q})$ is not a
      rational reconstruction of $A$ modulo $x^{D+1}$, then there is a
      non-zero constant $c$ such that
      $P_v = c \cdot U \left( \frac{y}{x} \right)
      x^{\max(\deg(U),\deg(R)+1)}$ (\Cref{prop:polInKernel}). In
      particular, $$N_1 = \max(\deg(U) - 1, \deg(R)).$$
    \end{corollary}

    \begin{proof}
      By \Cref{thm:rationalReconstrOfLowDegree}, there is a
      $Q \in \K[x]$ such that $Q \cdot (U_v,R_v) = (U,R)$. By
      \Cref{thm:vectInKernelToRationalApprox}, $(U_v,R_v)$ is a
      rational reconstruction, and so $\deg(Q) = 0$. Hence,
      $N_1 + 1 = \deg(P_v) = \max(\deg(U), \linebreak \deg(R)+1)$ and
      $Q \cdot P_v(1,\frac{y}{x}) x^{N_1+1} = U(\frac{y}{x}) x^{N_1+1}$.
    \end{proof}

    If $(U,R)$ is a rational reconstruction of $A$ modulo $x^{D+1}$
    such that $\deg(U) + \deg(R) \leq D$ and $U(0) = 1$, then
    $\frac{R}{U}$ is the Pad\'e approximant of $A$ of type
    $(\deg(R), \deg(U))$ \cite[Section~7.1]{bostan2017algorithmes}.
    When this Pad\'e approximant exists, it is unique, meaning that
    for any rational reconstruction with this property the quotient
    $\frac{R}{U}$ is unique (we can invert $U \!\!\!\! \mod x^{D+1}$ because
    $U(0) = 1$).
    When $N_1 < N_2$, we have that $\frac{D + 1}{2} \leq N_2$
    (\Cref{prop:existenceN1andN2}) and so, if the the Pad\'e
    approximant of $A$ of type $(\frac{D+1}{2} - 1, \frac{D+1}{2})$
    exists, by \Cref{thm:rationalReconstrOfLowDegree}, we can recover
    $P_v$ from it. The existence of this Pad\'e approximant is
    equivalent to the condition $U_v\!(0) = 1$, which means
    $v_{N_1+1} = 1$.
    In the algorithm proposed in the conference version of this paper
    \cite[Algorithm~3]{bender2016superfast}, the correctness of our
    algorithms relied on this condition. In that version, we ensured
    this property with a generic linear change of coordinates in the
    original polynomial $f$.
    In this paper, we skip this assumption. Following
    \citet[Theorem~7.2]{bostan2017algorithmes}, when $N_1 < N_2$, we can
    compute $v$ no matter the value of $v_{N_1+1}$. This approach has
    a softly-linear arithmetic complexity and involves the computation
    of a row of the $\egcd$ of $A$ and $x^{D+1}$. We can compute $P_w$
    from a consecutive row.

    Before going into the proof, we study the case $N_1 = N_2$. In
    this case, there are not rational reconstructions with the
    prerequisites of \Cref{thm:rationalReconstrOfLowDegree}, and so we
    treat this case in a different way.

    \begin{lemma} \label{thm:computingVwhenN1equalsN2}

      If $N_1 = N_2$, then there is a unique rational decomposition $(U,R)$
      such that $\deg(U) \leq \frac{D}{2}$, $\deg(R) \leq \frac{D}{2}$
      and $R$ is monic. In particular, $\deg(R) = \frac{D}{2}$ and we
      can consider the kernel polynomial related to
      $v \in \ker(H^{N_1+1})$ (\Cref{prop:defVandW}) as
      $P_v = U\!\left( \frac{y}{x} \right) x^{\frac{D}{2}+1}$.
    \end{lemma}

    \begin{proof}
      First note that, as $D = N_1 + N_2$
      (\Cref{prop:existenceN1andN2}), then $N_1 = \frac{D}{2}$.
      Following \Cref{eq:polMultiplication}, if we write
      $U = \sum_{i = 0}^{N_1} u_i x^i$ and
      $R = \sum_{i = 0}^{N_1} u_i x^i$, then we get the linear system,
      \begin{align*}
        %%%=== Bracket
        \begin{array}{c}
        \vphantom{
        \begin{pmatrix}
                 &         &            & a_{0}       \\
                 &         & a_0        & a_{1}       \\
                 & \iddots & \iddots    & \vdots     
               \end{pmatrix}
        } \\
        H^{N_1} \left\{
        \vphantom{
        \begin{pmatrix}
          a_0    & \cdots & a_{N_1-1} & a_{N_1}   \\
          a_1    & \cdots & a_{N_1} & a_{N_1+1}   \\
          \vdots & \iddots & \vdots  & \vdots  \\
          a_{D-N_1} & \cdots & a_{D-1} & a_{D} \\
        \end{pmatrix} 
        } \right.
        \end{array}
        %%%===
        \begin{pmatrix}
                 &         &            & a_{0}       \\
                 &         & a_0        & a_{1}       \\
                 & \iddots & \iddots    & \vdots      \\ \hline
          a_0    & \cdots & a_{N_1-1} & a_{N_1}   \\
          a_1    & \cdots & a_{N_1} & a_{N_1+1}   \\
          \vdots & \iddots & \vdots  & \vdots  \\
          a_{D-N_1} & \cdots & a_{D-1} & a_{D} \\
        \end{pmatrix} 
        \begin{pmatrix}
          u_{N_1} \\
          \vdots \\
          u_1 \\
          u_0 \\
        \end{pmatrix}
        =
        \begin{pmatrix}
          r_{0} \\
          % r_{1} \\
          \vdots \\
          r_{N_1 - 1} \\ \hline
          r_{N_1} \\
          0 \\
          \vdots \\
          0
        \end{pmatrix} .
      \end{align*}
      The matrix
      $H^{N_1} \in {\K}^{(D-\frac{D}{2}+1) \times (\frac{D}{2}+1)}$ is
      square and, as $\ker(H^{N_1}) = 0$ (\Cref{prop:defVandW}), it is
      invertible. If $r_{N_1} = 0$, that is $\deg(R_v) < N_1$, then
      the polynomial $U$ is zero and so $(U,R)$ is not a rational
      reconstruction. Hence, we can consider $\deg(U) = N_1$. If $R$
      is monic, then $r_{N_1} = 1$ and so we compute the coefficients
      of $U$ and $R$ as
      \begin{align*}
      \begin{pmatrix}
        u_{N_1} \\
        \vdots \\
        u_1 \\
        u_0 \\
        \end{pmatrix}
          =
          (H^{N_1+1})^{-1}
          \begin{pmatrix}
            1 \\
            0 \\
            \vdots \\
            0
          \end{pmatrix}
        %\quad
        , \quad
          \begin{pmatrix}
            r_{0} \\
            \vdots \\
            r_{N_1-1} \\
            1 
          \end{pmatrix}
          =
          \begin{pmatrix}
            &         &            & a_{0}       \\
            &         & a_0        & a_{1}       \\
            & \iddots & \iddots    & \vdots      \\
            a_0    & \cdots & a_{N_1-1} & a_{N_1}   \\
          \end{pmatrix} 
          \begin{pmatrix}
            u_{N_1} \\
            u_{N_1 - 1} \\
            \vdots \\
            u_0
          \end{pmatrix} .
          % (H^{N_1})^{-1} 
          % \begin{pmatrix}
          %   1 \\
          %   0 \\
          %   \vdots \\
          %   0 \\
          % \end{pmatrix}
      \end{align*}
    \end{proof}

    \begin{lemma}[{Correctness of \Cref{alg:gettingVandW}}] \label{thm:correctnessPvAndPw}
      Let $\{(U_j,V_j,R_j)\}_j$ be the set of triplets obtained from
      the Extended Euclidean Algorithm for the polynomials $A$ and
      $x^{D+1}$, see \Cref{sec:great-comm-divis}.  Let $i$ be the
      index of the first row of the extended Euclidean algorithm such
      that $\deg(R_i) < \frac{D+1}{2}$. Then, we can compute the
      polynomials $P_v$ and $P_w$ of \Cref{prop:polInKernel} as
      \begin{enumerate}[label=\textbf{(\Alph*)}]
      \item $P_v = U_i(\frac{x}{y}) \cdot x^{\max(\deg(U_i), \deg(R_i) + 1)}$.
      \item If $\deg(R_i) > \deg(U_i)$,
        $P_w = U_{i+1}(\frac{x}{y}) \cdot x^{\deg(U_{i+1})}$.
      \item If $\deg(R_i) \leq \deg(U_i)$,
        $P_w = U_{i-1}(\frac{x}{y}) \cdot x^{\deg(R_{i-1}+1)}$.
      \end{enumerate}
    \end{lemma}

    \begin{proof}
      \textbf{(A).}  First observe that if $i$ is the first index such
      that the degree of $R_i$ is strictly smaller than
      $\frac{D+1}{2}$, then, by \Cref{thm:RiIsDecreasing}, the degree
      of $R_{i-1}$ has to be bigger or equal to
      $\frac{D+1}{2}$. Hence, the degree of $U_i$ is smaller or equal
      to $\frac{D+1}{2}$, because by \Cref{thm:relationOfDegreesEGCD},
      $\deg(U_i) = D + 1 - \deg(R_{i-1}) \leq D + 1 - \frac{D+1}{2} =
      \frac{D+1}{2}$.
      We can consider $R_{i-1}$, that is $i$ is strictly bigger than
      $0$, because the degree of $R_0 = x^{D+1}$ is strictly bigger
      than $\frac{D+1}{2}$.
      
      If $N_1 = N_2$, then $D$ is even and $N_1 = \frac{D}{2}$
      (\Cref{prop:existenceN1andN2}). As
      $\lfloor\frac{D+1}{2}\rfloor = \frac{D}{2}$,
      $\deg(R_i) \leq \frac{D}{2}$ and $\deg(U_i) \leq
      \frac{D}{2}$. By \Cref{thm:computingVwhenN1equalsN2},
      $\max(\deg(U_i),\deg(R_i)+1) = N_1+1$ and we can consider $P_v$
      as $U_i(\frac{y}{x}) x^{N_1+1}$.
      
      If $N_1 < N_2$, assume that there is a non-zero $Q \in \K[x]$
      such that $Q$ divides $U_i$ and $R_i$ and
      $(\frac{U_i}{Q}, \frac{R_i}{Q})$ is a rational reconstruction of
      $A$ modulo $x^{D+1}$. Then,
      $\frac{U_i}{Q} A \equiv \frac{R_i}{Q} \mod x^{D+1}$ and so there
      is a polynomial $\bar{V}$ such that
      $\bar{V} x^{D+1} + \frac{U_i}{Q} A = \frac{R_i}{Q}$. Multiplying
      by $Q$, we get the equality
      $Q \, \bar{V} x^{D+1} + U_i A = R_i$.  Consider the identity
      $V_i x^{D+1} + U_i A = R_i$ from
      \Cref{eq:bezoutIdentity}. Coupling the two equalities together,
      we conclude that $V_i = Q \bar{V}$. As $Q$ divides $U_i$ and
      $V_i$, which are coprime (\Cref{thm:bezoutFactorsCoprime}), $Q$
      is a constant, that is $\deg(Q) = 0$.
      If $N_1 < N_2$, then $D < 2 N_2$ (\Cref{prop:existenceN1andN2}).
      Hence, $\deg(U_i) \leq \frac{D+1}{2} \leq N_2$ and
      $\deg(R_i) + 1 < \frac{D+1}{2} + 1 \leq N_2 + 1$, that is,
      $\max(\deg(U_i), \deg(R_i) + 1) \leq N_2$.
      Hence, by \Cref{thm:rationalReconstrOfLowDegree}, we can
      consider $U_i(\frac{y}{x}) x^{\max(\deg(U_i), \deg(R_i) + 1)}$
      as the kernel polynomial $P_v$ from \Cref{prop:polInKernel},
      spanning $\ker(H^{N_1+1})$.

      \textbf{(B).}  Assume that the degree of $U_i$ is strictly
      bigger than the one of $R_i$, that is $\deg(U_i) >
      \deg(R_i)$. Then $N_1 = \deg(U_i) - 1$, as
      $\deg(U_i) = \deg(P_v) = N_1+1$
      (\Cref{thm:degreeKernelPol}). Note that in this case $i>1$
      because $U_1 = 1$, $R_1 = A \neq 0$, and so
      $\deg(U_1) \leq \deg(R_1)$.
      The degree of $R_{i-1}$ is $N_2$ because, by
      \Cref{thm:relationOfDegreesEGCD},
      $\deg(R_{i-1}) = D+1 - \deg(U_i) = D+1 -N_1 - 1 = N_2$ 
      (\Cref{prop:existenceN1andN2}).
      Consider the degree of $U_{i-1}$. By
      \Cref{thm:relationOfDegreesEGCD},
      $\deg(U_{i-1}) = D+1 - \deg(R_{i-2})$. As
      $\deg(R_{i-2}) > \deg(R_{i-1})$ (\Cref{thm:RiIsDecreasing}),
      then
      $\deg(R_{i-2}) > N_2$. Therefore, the degree of $U_{i-1}$ is
      smaller or equal to the one of $R_{i-1}$ because
    $$\deg(U_{i-1}) = D+1 - \deg(R_{i-2}) < D + 1 - N_2 = N_1 + 1 \text{, and so}$$
    $$\deg(U_{i-1}) \leq N_1 \leq N_2 = \deg(R_{i-1}).$$
    Hence, by \Cref{thm:egcdGivesRatRecontrs}, $(U_{i-1}, R_{i-1})$ is
    a rational reconstruction of $A$ modulo $x^{D+1}$ such that
    $\deg(U_{i-1}) \leq N_1$ and $\deg(R_{i-1}) = N_2$.
    So, $\max(\deg(U_{i-1}), \deg(R_{i-1}) + 1) = N_2 + 1$ and, by
    \Cref{thm:degreeKernelPol}, there is a kernel polynomial
    $P_\omega = U_{i-1}(\frac{y}{x}) x^{N_2+1}$ of degree $N_2+1$ such
    that $x^{N_2 + 1 - \deg(U_{i-1})}$ divides $P_\omega$. As
    $\deg(U_{i-1}) \leq N_1$, $x^{N_2 + 1 - N_1}$ divides
    $x^{N_2 + 1 - \deg(U_{i-1})}$ and so, it divides $P_\omega$. We
    assumed that the degree of $U_i$ is strictly bigger than the one
    of $R_i$, and so $x$ does not divide $P_v$
    (\Cref{thm:degreeKernelPol}). Hence, there is no binary form $Q$
    of degree $N_2 - N_1$ such that $x^{N_2 - N_1 + 1}$ divides
    $Q \, P_v$. Therefore, by \Cref{thm:PwNotInPv}, we can consider
    $P_w = P_\omega$.

    \textbf{(C).}  Assume that the degree of $R_i$ is bigger or equal
    to the one of $U_i$, that is $\deg(R_i) \geq \deg(U_i)$. Hence,
    $\deg(R_i) + 1 = \deg(P_v) = N_1 + 1$
    (\Cref{thm:degreeKernelPol}), and so $\deg(R_i) = N_1$.  In
    particular, $R_i \neq 0$, and so the $(i+1)$-th row of the
    Extended Euclidean Algorithm, $(U_{i+1},V_{i+1},R_{i+1})$, is
    defined.
    The degree of $U_{i+1}$ is $N_2+1$ because, by
    \Cref{thm:relationOfDegreesEGCD},
    $\deg(U_{i+1}) = D+1 - \deg(R_{i}) = N_2 + 1$
    (\Cref{prop:existenceN1andN2}).
    The degree of $R_{i+1}$ is strictly smaller than the one of $R_i$
    (\Cref{thm:RiIsDecreasing}), which is $N_1$. Hence, the degree of
    $R_{i+1}$ is smaller than the degree of $U_{i+1}$ because
    $\deg(R_{i+1}) < N_1 \leq N_2 < \deg(U_{i+1})$.
    Therefore, $P_\omega = U_{i+1}(\frac{y}{x}) x^{N_2 + 1}$ is a
    kernel polynomial in $\ker(H^{N_2+1})$
    (\Cref{thm:fromRationalToKernelPolynomial}).  By
    \Cref{thm:degreeKernelPol}, as $\deg(R_{i+1}) < \deg(U_{i+1})$,
    $x$ does not divide $P_\omega$. Also, the maximal power of $x$
    that divides $P_v$ is $x^{\deg(R_{i}) + 1 - \deg(U_{i})}$, and, as
    we assumed $\deg(R_{i}) \geq \deg(U_{i})$, $x$ divides
    $P_v$. Hence, every polynomial in
    $\{Q \, P_v : \deg(Q) = N_2 - N_1\}$ is divisible by $x$, and so,
    by \Cref{thm:PwNotInPv}, we can consider $P_w = P_\omega$.
  \end{proof}

  \begin{algorithm}[h]
  \caption{\textsc{Compute\_Pv\_and\_Pw}} 
  \label{alg:gettingVandW}           
  \begin{algorithmic}      

    \REQUIRE A sequence $(a_0,\dots,a_D)$.

    \ENSURE Polynomials $P_v$ and $P_w$ as \ref{prop:polInKernel}.
    
    \noindent
    \begin{enumerate}[leftmargin=*]
    \item $i \leftarrow$ first row of
      $\egcd(x^{D+1}, \sum_{i=0} a_i x^i)$ such that
      $R_i < \frac{D + 1}{2}$.

    \item $P_v \leftarrow U_i(\frac{x}{y}) \cdot x^{\max(\deg(U_i), \deg(R_i) + 1)}$.
      
      $N_1 \leftarrow {\max(\deg(U_i)-1, \deg(R_i))}$

      \item \textbf{IF} $\deg(R_i) > \deg(U_i)$,

        \qquad $P_w \leftarrow U_{i+1}(\frac{x}{y}) \cdot x^{\deg(U_{i+1})}$.

        \qquad $N_2 \leftarrow \deg(U_{i+1}) - 1$.

      \item[] \textbf{ELSE}

        \qquad  $P_w \leftarrow U_{i-1}(\frac{x}{y}) \cdot x^{\deg(R_{i-1}+1)}$.      

        \qquad $N_2 \leftarrow \deg(R_{i-1})$.

      \item \textbf{Return} $P_v$ and $P_w$

    \end{enumerate}
  \end{algorithmic}
  \end{algorithm}

  \subsection[Computing a square-free polynomial Q]
  {Computing a square-free polynomial $Q$} 
  \label{sub:computing_a_square_free}

  We can compute $Q$ at \cref{stp:choosingQ} of
    \Cref{alg:algorithmDecomp} in different ways. If
    $P_v$ is square-free, then we set $Q$ equal to $P_v$.
    If $P_v$ is not square-free, by \Cref{lem:rankOfDecomposition}, we
    need to find a vector $\mu \in \field^{(N_2-N_1+1)}$ such that
    $Q_{\mu} := P_{\mu}\*P_v + P_w$ is square-free.
    By \Cref{prop:PvPwCoprimes}, $P_v$ and $P_w$ are relative prime.
    Thus, if we take a random vector $\mu$, generically, $Q_\mu$ would
    be square-free. For this to hold, we have to prove that the
    discriminant of $Q_{\mu}$ is not identically zero. To simplify
    notation, in the following theorem we dehomogenize the
    polynomials.

    \begin{theorem}
      \label{theo:nonZeroDiscriminant}
      Given two relative prime univariate polynomials $P_v(x)$ and
      $P_w(x)$ of degrees $N_1+1$ and $N_2+1$ respectively, let 
      $Q_\mu(x) := P_{\mu}(x,1) \*P_v + P_w \in
      \field[\mu_0,\dots,\mu_{N_2-N_1}][x]$. The discriminant of
      $Q_\mu(x)$ with respect to $x$ is a non-zero polynomial.
    \end{theorem}

    \begin{proof}
      The zeros of the discriminant of $Q_\mu(x)$ with respect to $x$
      over $\K$ correspond to the set \linebreak
      $\{\mu \in \field^{N_2-N_1+1} : Q_\mu \text{ has double
        roots}\}$. We want to prove that the discriminant is not zero.

      % The zeros in $\field$
      % of the discriminant of
      % $Q_\mu(x)$, if we consider it as a univariate polynomial in $x$,
      % correspond to the set
      % $\{\mu \in \field^{N_2-N_1+1} : Q_\mu \text{
      % has double roots}\}$.

      A univariate polynomial is square-free if and only if it does
      share any root with its derivative.  Hence,
      $(\mu_0, \dots, \mu_{N_2-N_1})^\transpose \in \{\mu \in
      \field^{N_2-N_1+1} : Q_\mu \text{ has double roots}\}$ if and
      only if, there is
      $(\mu_0, \dots, \mu_{N_2-N_1},\alpha) \in \field^{N_2-N_1+1}
      \times \closurefield$ such that it is a solutions to the
      following system of equations
      \begin{align}
      \label{eq:doubleroots}
        \begin{cases}
          (P_{{\mu}} \* P_v + P_{w})(x) = 0 \\
          (P_{{\mu}} \* P'_v + P'_{{\mu}} \* P_v  + P'_{w})(x) = 0 .
        \end{cases} 
      \end{align}

      In \Cref{eq:doubleroots}, $\mu_0$ only appears in $P_\mu(x,1)$
      with degree $1$. We eliminate it to obtain the polynomial
      $$( P_v \cdot P_{\mu}' + P_w') P_v - P_v' \cdot P_w.$$
      
      % We can rewrite $P_{\mu}$ as $\mu_0 + x P_{\hat{\mu}}$, where
      % $\hat{\mu} := (\mu_1, \dots, \mu_{N_2-N_1})^\transpose$. From
      % \Cref{eq:doubleroots}, we can eliminate $\mu_0$ to obtain the
      % equation $ \big( P_v (P_{\hat{\mu}} \* P_v + P_{w})' - P'_v
      % (P_{\hat{\mu}} \* P_v + P_{w}) \big) (x) = 0$.
      % %%

        This polynomial is not identically $0$ as $P_v'$ does not
        divide $P_v$ and $P_v$ and $P_w$ are relative prime. Hence,
        for each $(\mu_1,\dots,\mu_{N_2-N_1})$, there is a finite
        number of solutions for this equation, bounded by the degree
        of the polynomial. Moreover, as the polynomials of
        \Cref{eq:doubleroots} are linear in $\mu_0$, each solution of
        the deduced equation is extensible to a finite number of
        solutions of \Cref{eq:doubleroots}.  Hence, there is a
        $\mu \in \field^{N_2-N_1+1}$, such that $Q_{\mu}$ is
        square-free. Therefore, the discriminant of $Q_{\mu}(x)$ is
        not identically zero.
    \end{proof}

    \begin{corollary}
      \label{lem:generalQsqFree}
      For every vector
      $(\mu_1, \dots, \mu_{N_2-N_1}) \in \field^{N_2-N_1}$ such that
      there is a $\mu_0 \in \K$ such that $y^2$ does not divides $Q_\mu$,
      where $\mu = (\mu_0,\dots,\mu_{N_2-N_1})$, there are at most
      $2 D + 2$ different values for $\mu_0 \in \field$ such that the
      polynomial $Q_{\mu}(x,y)$ is not square-free.
    \end{corollary}

    \begin{proof}
      If $Q_\mu(x,y)$ is not square-free, then it has a double root in
      $\PC$. This root could be of the form $(\alpha,1) \in \PC$ or
      $(1,0) \in \PC$. We analyze separately these cases.

      First, we consider the polynomial $Q_\mu(x,1) \in
      \K[\mu_0,x]$. By \Cref{theo:nonZeroDiscriminant}, the
      discriminant of $Q_\mu(x,1)$ with respect to $x$ is not zero. As
      $Q_\mu(x,1)$ is a polynomial of degree $N_2+1$ with respect to
      $x$, and of degree $1$ with respect to $\mu_0$, the degree with
      respect to $\mu_0$ of the discriminant of $Q_\mu(x,1)$ with
      respect to $x$ is at most $(N_2 + 1) + N_2 \leq 2 D + 1$.
      Hence, there are at most $2 D + 1$ values for $\mu_0$ such that
      $Q_{\mu}(x,y)$ has a root of the form $(\alpha,1) \in \PC$ with
      multiplicity bigger than one.

      The polynomial $Q_{\mu}(x,y)$ has a root of the form
      $(1,0) \in \PK$ with multiplicity bigger than one, if and only
      if $y^2$ divides $Q_{\mu}(x,y)$. If this happens, then the
      coefficients of the monomials $y \cdot x^{N_2-N_1-1}$ and
      $x^{N_2-N_1}$ in the polynomial $Q_{\mu}(x,y)$ are zero. By
      assumption, these coefficients are not identically zero as
      polynomials in $\K[\mu_0]$. As $Q_{\mu}(x,y)$ is a linear polynomial
      with respect to $\mu_0$, there is at most one value for $\mu_0$
      such that $y^2$ divides $Q_{\mu}(x,y)$.

      Therefore, there are at most $(2 D + 1) + 1$ values such
      that $Q_{\mu}(x,y)$ is not square-free.
    \end{proof}

    \begin{remark}
      The previous assumption is not restrictive. If $y^2$ divides
      $Q_\mu$, where $\mu = \linebreak (\mu_0,\dots,\mu_{N_2-N_1})$, then $y^2$
      does not divide
      $Q_{(\mu_0,\dots,\mu_{N_2-N_1}+1)} = Q_\mu + x^{N_2+1}$ nor \linebreak
      $Q_{(\mu_0,\dots,\mu_{N_2-N_1-1}+1, \mu_{N_2-N_1})} = Q_\mu + y
      x^{N_2}$. Moreover, if $N_2-N_1 \geq 2$, $y^2$ divides (or not)
      $Q_\mu(x,y)$ regardless the value of $\mu_0$.  Conversely, if
      $N_2-N_1 < 2$, there is always a $\mu_0$ such that $y^2$ does
      not divide $Q_\mu$.
    \end{remark}

  \subsection{Correctness of Algorithm \ref{alg:algorithmDecomp}} 
  \label{sub:correctness_of_algorithm_decomposition}

  For computing a decomposition for a binary form $f$, we need to
  compute the kernel of a Hankel matrix (\Cref{theo:Silv1851}).
  \Cref{alg:gettingVandW} computes correctly the polynomials $P_v$ and
  $P_w$ that characterize the kernels of the family of Hankel matrices
  associated to $f$. Once we obtain these polynomials,
  \cref{stp:choosingQ} (see~\Cref{lem:generalQsqFree}) and
  \cref{stp:solveLambdas} computes the coefficients
  $\alpha_j, \beta_j, \lambda_j$ of the decomposition. Hence, we have
  a decomposition for $f$, as $f(x,y) = \sum_{j = 1}^r \lambda_j \* ( \alpha x + \beta y )^D.$

  \paragraph*{Example} 
    \label{sub:example}
    Consider $f(x,y) = y^4 + 8 x y^3 + 18 x^2 y^2 + 16 x^3 y + 5 x^4$.
    The family of Hankel matrices  associated to $f$ are related to the vector
    $a := (1,2,3,4,5)^\transpose$, it is denoted by $\{H_a^{k}\}_k$,  and it contains the following matrices:
    \begin{center}
%    \vspace{10px}
    \begin{tabular}{ | c | c | c | c | }
      \hline & & & \\
        $H^1_a = \begin{pmatrix} 1 & 2 \\ 2 & 3 \\ 3 & 4  \\  4 & 5 \end{pmatrix}$  &
        $H^2_a = \begin{pmatrix} 1 & 2 & 3 \\ 2 & 3 & 4 \\ 3 & 4 & 5 \end{pmatrix}$  &
        $H^3_a = \begin{pmatrix} 1 & 2 & 3 & 4 \\ 2 & 3 & 4 & 5 \end{pmatrix}$ &
        $H^4_a = \begin{pmatrix} 1 & 2 & 3 & 4 & 5 \end{pmatrix}$ \\
       & & & \\ \hline
    \end{tabular}
    \end{center}
 %   \vspace{10px}

    The kernel $H^1_a$ is trivial, so we compute the one of
    $H^2_a$. This kernel is generated by the vector
    $(1,-2,1)^\transpose$, so by \Cref{prop:defVandW} we consider
    $v = (1,-2,1)^\transpose$. Also, by \Cref{prop:existenceN1andN2},
    $N_1 + 1 =2$ and $N_2 = D - N_1 = 3$. The kernel polynomial
    $P_v = y^2-2 x \, y + x^2 = (x-y)^2$ is not square-free thus, by
    \Cref{lem:rankOfDecomposition}, the rank of $f(x,y)$ is
    $N_2+1 = 4$ and we have to compute the kernel polynomial $P_w$ in
    the kernel of $H^4_a$.  Following \Cref{prop:defVandW}, the kernel
    of $H^4_a$ is generated by U-chain of $v$ given vectors
    $\Uch_2^0v = (1,-2,1,0,0)^\transpose$,
    $\Uch_2^1v = (0,1,-2,1,0)^\transpose$, and
    $\Uch_2^2v = (0,0,1,-2,1)^\transpose$, plus a vector $w$ linear
    independent with this U-chain. We consider the vector
    $w = (0,0,0,5,-4)$, which fulfills that assumption.
    Hence,
    $P_v = y^2-2 x \, y + x^2$ and $P_w = 5 y x^3 - 4 x^4$.
  
    We proceed by computing a square-free polynomial combination of
    $P_v$ and $P_w$. For that, we choose
    $$Q := (44 y^2 + 11 y x + 149 x^2) \, P_v + 36 P_w = (5 x - 11 y)(x-2 y)(x + 2 y)(x+y).$$

    Finally, we solve the system given by the transposed of a Vandermonde matrix,
    \begin{align}
      {
      \begin{pmatrix}
        5^4 & 1 & 1 & 1 \\
        11 \cdot 5^3 & 2 & -2 & -1 \\
        11^2 \cdot  5^2 & 2^2 & (-2)^2 & (-1)^2 \\
        11^3 \cdot  5 & 2^3 & (-2)^3 & (-1)^3
      \end{pmatrix}
                              \begin{pmatrix}
                                \lambda_1 \\ \lambda_2 \\ \lambda_3 \\ \lambda_4
                              \end{pmatrix}
      = 
      \begin{pmatrix}
        1 \\ 2 \\ 3 \\ 4
      \end{pmatrix}.
      }
    \end{align}
    The unique solution of the system is
    $(-\frac{1}{336},3,\frac{1}{21},\frac{3}{16})^\transpose$, and so
    we recover the decomposition
    $$f(x,y) =
    -\frac{1}{336}(11x+5y)^4+3(2x+y)^4+\frac{1}{21}
    (-2x+y)^4-\frac{3}{16}(-x+y)^4.$$

    Instead of considering incrementally the matrices in the Hankel
    family we can compute the polynomials $P_v$ and $P_w$ faster by
    applying \Cref{alg:gettingVandW}. For this, we consider the
    polynomial $A := 5 x^4+4 x^3+3 x^2+2 x+1$, and the rows of the
    Extended Euclidean Algorithm for $A$ and $x^{5}$.
    \begin{center}
      \begin{tabular}{ | c | c | c | c | }
        \hline
        $j$ & $V_j$ & $U_j$ & $R_j$\\ \hline
        $0$ 	 & 	 $1$ 	 & 	 $0$ 	 & 	 $x^5$ \\ 
        $1$ 	 & 	 $0$ 	 & 	 $1$ 	 & 	 $5 x^4+4 x^3+3 x^2+2 x+1$ \\ 
        $2$ 	 & 	 $1$ 	 & 	 $\frac{1}{25} (5x-4)$ 	 &   $\frac{1}{25} (x^3+2 x^2+3 x+4)$ \\ 
        $3$ 	 & 	 $-25 (5 x- 6)$ 	 & 	 $25 (x-1)^2$ 	 & 	 $25$ \\ 
        $4$ 	 & 	 $\frac{1}{25} (5 x^4+4 x^3+3 x^2+2 x+1)$ 	 & 	 $-\frac{1}{25} x^5$ 	 & 	 $0$ \\\hline
      \end{tabular}
    \end{center}
    
    We need to consider the first $j$ such that
    $\deg(R_j) < \frac{5}{2}$, which is $j = 3$.  Hence,
    \linebreak
    $N_1 = \max(\deg(U_3) - 1, \deg(R_3)) = 1$ and
    $$P_v := U_3\!\!\left(\frac{y}{x}\right) \, x^{\max(\deg(U_3), \deg(R_3) + 1)} =
    25 \left(\frac{y}{x}-1 \right)^2 x^2 = 25 (y - x)^2.$$
    As $\deg(R_3) \leq \deg(U_3)$, we consider $N_2 = \deg(R_2)  =3$ and
    
    \vspace{\abovedisplayskip}
    \hfill $\displaystyle P_w := U_2\!\!\left(\frac{y}{x}\right) \, x^{\deg(R_2) + 1} =
    \frac{1}{25} (5y x^3 - 4x^4).$ \hfill $\Diamond$

     % $$H^{2} \cdot v = \begin{pmatrix} 1 & 2 & 3 \\ 2 & 3 & 4 \\ 3 & 4 & 5 \end{pmatrix} \begin{pmatrix} 25 \\ -50 \\ 25 \end{pmatrix} = 0, \text{ and}$$

     % $$
     % H^{4} \cdot \begin{pmatrix} \Uch^0_{2} v & \Uch^1_{2} v & \Uch^2_{2} v & w \end{pmatrix} = 
     % \begin{pmatrix} 1 & 2 & 3 & 4 & 5 \end{pmatrix}
     %  \begin{pmatrix} 25 & 0 & 0 & 0 \\ -50 & 25 & 0 & 0 \\ 25 & -50 & 25 & 0 \\ 0 & 25 & -50 & \frac{1}{5} \\ 0 & 0 & 25 & \frac{-4}{25} \end{pmatrix} = 0 $$

    % As $P_{v}$ is not square free, the rank of $f$ is $N_2+1= \deg(P_w) = 4$ and we
    % need to get a square-free kernel polynomial of degree $4$. We consider
    % the square-free kernel polynomial
    % {\small
    %   $$\left(\frac{44}{25} y^2+\frac{112}{25} x y+\frac{149}{25} x^2\right) P_v - 900  Pw =
    %   -900
    %   Q_{\left(\tfrac{-11}{5625},\tfrac{-28}{5625},\tfrac{-149}{22500}\right)^\transpose}
    %   = \left(5 x - 11 y\right)\left(x-2 y\right)\left(x + 2 y\right)\left(x+y\right).$$
    % }
    % Solving the system of \Cref{eq:gettingLambdas}, we obtain the
    % decomposition
    
    % \vspace{\abovedisplayskip}
    % \hfill $\displaystyle f(x,y) = -\frac{1}{336}(11x+5y)^4+3(2x+y)^4+\frac{1}{21}
    % (-2x+y)^4-\frac{3}{16}(-x+y)^4.$ \hfill $\Diamond$

    \paragraph{The real case}
    When we consider the decomposition of binary forms over $\R$,
    Algorithm \ref{alg:algorithmDecomp} might fail. This happens
    either when the decomposition over $\C$ is not unique or when the
    decomposition over $\C$ is unique but it is not a decomposition
    over $\R$. The algorithm fails because 
    \begin{itemize}
    \item The real rank of the binary form might be bigger than
      $N_2+1$, see
      \cite{reznick2013length}. \Cref{lem:rankOfDecomposition} does
      not hold over $\R$ and so we cannot find a square-free kernel
      polynomial $Q_\mu$ that factors over $\R$.
    \item Even when the real rank of the binary form is $N_2+1$, we
      need to perform some extra computations to compute a square-free
      kernel polynomial $Q_\mu$ that factors over $\R$. This
      computations are not taken into account in our algorithm, so we
      could never find such a decomposition.
    \end{itemize}
    Recently, Sylvester's algorithm was extended to the real case
    \cite[Algorithm~2]{ansola_semialgebraic_2017}. This algorithm
    performs an incremental search over $r$ as in \Cref{alg:common},
    but it decides if there is a real decomposition of length $r$ by
    checking the emptiness of a semi-algebraic set. In the step $r$-th
    of the algorithm, the semi-algebraic set is embedded in
    $\R^{\dim(\ker(H^r))}$. Hence, the  bottleneck of
    their algorithm is not the computation of $P_v$ and $P_w$ as in
    our case, but deciding the emptiness of the semi-algebraic set.
    We emphasize that when the decompositions over $\C$ and $\R$ are
    unique and the same, our algorithm computes such a
    decomposition. Moreover, given $P_v$, we can check if the previous
    condition holds by checking if $P_v$ splits over $\R$.    
    
\section{Complexity} 
\label{sec:complexity}

\noindent In this section we study the algebraic degree of the
parameters that appear in the decomposition of a binary form as well
as the arithmetic and bit complexity of \Cref{alg:algorithmDecomp}.

\subsection{Algebraic degree of the problem}
\label{sec:algebraicDegree}

If we assume that the coefficients of the input binary form
\Cref{eq:bf} are rational numbers then the parameters of the
decompositions, $\alpha_j$, $\beta_j$, and $\lambda_j$ (see
\Cref{eq:bfDecomp}), are algebraic numbers, that is, roots of
univariate polynomials with integer coefficients. The maximum degree
of these polynomials is the algebraic degree of the problem. We refer
the interested reader to
\cite{bajaj1988algebraic,nie2010algebraic,dhost-dgl-16} for a detailed
exposition about the algebraic degree and how it address the
complexity of the problem at hand at a fundamental level.

\subsubsection{The complexity of computing Q} 
\label{sec:getting_mu}

Recall that, from \Cref{lem:rankOfDecomposition}, the rank of $f$ could
be either $N_1+1$ or $N_2 + 1$.  When the polynomial $P_v$ is
square-free, then the rank is $N_1 + 1$ and $Q = P_v$.
Following the discussion of \Cref{sub:computing_a_square_free}, we
prove that, when the rank of the binary form is $N_2+1$, we can
compute a square-free kernel polynomial $Q$ of this degree such that
the largest degree of its irreducible factors is $N_1$. Moreover, we
prove that for almost all the choices of $(N_2-N_1+1)$ different
points in $\PK$ (the projective space of $\K$) there is a square-free
kernel polynomial of $H^{N_2+1}$ which vanishes on these points. This
will be our choice for $Q$.

\begin{lemma}
  \label{lem:iterpolMu}
  Let $f$ be a binary form of rank $N_2+1$. Given $(N_2-N_1+1)$
  different points
  $(\alpha_0,\beta_0),\dots,(\alpha_{N_2-N_1},\beta_{N_2-N_1}) \in
  \PK$
  % %% 
  % \footnote{ The set of forms $\{(\beta_i x - \alpha_i y)\}_i$ are
  %   pairwise linearly independent if every point of the form
  %   $(\alpha_i,\beta_i) \in \PP(\closurefield)$ is different,
  %   where $\PP(\closurefield)$ is the projective space of
  %   $\closurefield$.}
  % %% 
  such that none of them is a root of $P_v$, then there is a unique
  binary form $P_\mu$ of degree $N_2-N_1$, such that the kernel
  polynomial $Q_\mu := P_\mu \* P_v + P_w$ vanishes on these points.
\end{lemma}
    
    \begin{proof}
      Without loss of generality, we assume $\beta_i = 1$, for
      $i \in \{1,\dots,N_2-N_1\}$. By \Cref{prop:polInKernel}, for any
      polynomial $P_\mu$ of degree $N_2-N_1$, $Q_\mu$ is a kernel
      polynomial. Since $Q_\mu(\alpha_i,1) = 0$, we can interpolate
      $P_\mu$ by noticing that
      $P_\mu(\alpha_j,1) =
      -\,\frac{P_w(\alpha_j,1)}{P_v(\alpha_j,1)}$.

      The degree of $P_\mu$ is $(N_2-N_1)$ and we interpolate it using
      $(N_2-N_1+1)$ different points. Hence there is a unique
      interpolation polynomial $P_\mu$. So, $Q_\mu$ is the unique
      kernel polynomial of $H^{N_2+1}$ vanishing at all these points.
    \end{proof}

    \noindent\paragraph*{Example (cont.)}
    For the example of \Cref{sub:example}, we obtained the square-free
    kernel polynomial $Q$ by choosing the points $(2,1)$, $(-2,1)$ and
    $(-1,1) \in \PK$. If we choose other points such that $Q$ is
    square-free, we will obtain a different decomposition. Hence, $f$
    does not have a unique decomposition.  This holds in general.
    \hfill $\Diamond$

    From \Cref{lem:iterpolMu}, we deduce the following well-known
    result about the uniqueness of the decomposition, see also
    \cite{helmke1992waring,comas2011rank,bernardi2011computing,carlini_waring_2017}.
    
    \begin{corollary} \label{thm:decompIsUnique}
      A decomposition is unique if and only if the rank is $N_1 + 1$
      and $N_1 < N_2$.
      A decomposition is not unique if and only if the rank is
      {$N_2+1$}.
    \end{corollary}
    
    \begin{theorem}
      \label{theo:boundOfAlgDegree}
      Let the rank of $f$ be $N_2+1$. Then there is a square-free kernel
      polynomial $Q$ such that the largest degree of its irreducible
      factors is at most $N_1$.
    \end{theorem}
    
    \begin{proof}
      If the rank of $f$ is $N_2+1$,
      then for each set of $N_2-N_1+1$ different points in $\PC$,
      following the assumptions of \Cref{lem:iterpolMu}, 
      there is a unique kernel polynomial. 
      There is a rational map that realizes this relation
      (see the proof of \Cref{lem:iterpolMu}).
      Let this map be $Q_{[\overline{\alpha}]}$, where \linebreak
      $\overline{\alpha} = \left((\alpha_0, \beta_0), \dots, (\alpha_{N_2-N_1}, \beta_{N_2-N_1})\right) \in
      \PC^{N_2-N_1+1}$.
      The image of the map is contained in
      $\{ P_\mu \* P_v + P_w : \mu \in \closurefield^{N_2 - N_1 + 1}
      \}$.  This set and $\PC^{N_2-N_1+1}$ have the same dimension,
      $N_2-N_1+1$.

      Given a kernel polynomial $\hat{Q}(x,y)$, there is a finite
      number of distinct points $(\alpha,\beta) \in \PC$ such that
      $\hat{Q}(\alpha,\beta) = 0$. Hence, the pre-image of an element
      in the image of $Q_{[\overline{\alpha}]}$ is a finite
      set. Therefore, the dimension of the image and the dimension of
      the domain are the same.

      By \Cref{theo:nonZeroDiscriminant}, the non-square-free kernel
      polynomials form a hypersurface in the space of kernel
      polynomials of the shape $P_\mu \* P_v + P_w$. If we consider
      the pre-image of the intersection between this hypersurface and
      the image of the rational map, then its dimension is smaller
      than $N_2-N_1+1$. 

      Therefore, generically, for $N_2-N_1+1$ different points in
      $\PK$, the map $Q_{[\overline{\alpha}]}(x,y)$ results a
      square-free kernel polynomial. As $\closurefield$ is the
      algebraic closure of $\field \subset \CC$, the same holds over
      $\field$.
    \end{proof}

    \begin{theorem}
      \label{cor:dg-irr-Q}
      Given a binary form $f$ of rank $r$ and degree $D$, there is a
      square-free kernel polynomial of degree $r$ such that
      the biggest degree of its irreducible factors is $\min(r, D-r+1)$.
    \end{theorem}

    \begin{proof}
      If the rank is $r = N_2+1$, then $\min(r, D-r+1) = N_1$.
      By \Cref{theo:boundOfAlgDegree}, such a square-free kernel
      polynomial exists. If the rank is $r = N_1+1$ and $N_1 < N_2$,
      by \Cref{lem:rankOfDecomposition}, there is a square-free kernel
      polynomial of degree $\min(r, D-r+1) = N_1+1$.
    \end{proof}

    The previous result is related to the decomposition of tensors of
    the same border rank \cite[Theorem 2]{comas2011rank};
    \cite[Theorem 23]{bernardi2011computing};
    \citep{blekherman2015typical}.

    We can also bound the number of possible bad choices in the proof
    of \Cref{theo:boundOfAlgDegree}.
    
    \begin{proposition} \label{thm:badValuesToInterpolateQ}
      Let $f$ be a binary form of rank $N_2+1$. For every set
      $S \subset \PK$ of cardinal $(N_2-N_1)$ such that
      $(\forall (\alpha,\beta) \in S) \; P_v(\alpha,\beta) \neq 0$ there
      are at most $D^2 + 3 D + 1$ values $(\alpha_0,\beta_0) \in \PK$
      such that $(\alpha_0,\beta_0) \not\in S$,
      $P_v(\alpha_0,\beta_0) \neq 0$ and the unique kernel polynomial
      $Q_\mu := P_\mu \* P_v + P_w$ that vanish over $S$ and
      $(\alpha_0,\beta_0)$ (\Cref{lem:iterpolMu}) is not square-free.
    \end{proposition}
    
    To prove this proposition we use Lagrange polynomials to construct
    the maps and varieties of the proof of
    \Cref{theo:boundOfAlgDegree}.

    Let
    $S =
    \left\{(\alpha_1,\beta_1),\dots,(\alpha_{N_2-N_1},\beta_{N_2-N_1})\right\}
    \subset \PK$ be the set of \Cref{thm:badValuesToInterpolateQ}. For
    each $(\alpha_0,\beta_0) \in \PK$ such that
    $(\alpha_0,\beta_0) \not\in S$ and
    $P_v(\alpha_0,\beta_0) \neq 0$ we consider the unique kernel
    polynomial $Q^{\alpha_0,\beta_0}$ which vanishes at $S$ and
    $(\alpha_0,\beta_0)$, see~\Cref{lem:iterpolMu}. Using Lagrange
    polynomial, we can write this polynomial as
    $$
    Q^{\alpha_0,\beta_0}(x,y) =
    \left(
      -\,\frac{P_w(\alpha_0,\beta_0)}{P_v(\alpha_0,\beta_0)} \frac{M(x,y)}{M(\alpha_0,\beta_0)}
      + \sum_{i = 1}^{N_2-N_1} \frac{\beta_0 x - \alpha_0 y}{\alpha_0 \beta_i - \alpha_i \beta_0}
      E_i(x,y)
    \right)
    P_v(x,y)
    + P_w(x,y)
    $$
    Where
    $E_i(x,y) := - \frac{P_w(\alpha_i,\beta_i)}{P_v(\alpha_i,\beta_i)}
    \prod_{j \not\in \{0,i\}} \frac{\beta_j x - \alpha_j y}{\alpha_i
      \beta_j - \alpha_j \beta_i}$ and
    $M(x,y) := \prod_{j = 1}^{N_2-N_1} (\beta_j x - \alpha_j
    y)$.\footnote{\, For each $0 \leq i \leq N_2-N_1$,
      $Q^{\alpha_0,\beta_0}(x,y)$ is a rational function of degree $0$
      with respect to $(\alpha_i,\beta_i)$. Hence, it is well defined the
      evaluation of the variables $(\alpha_i,\beta_i)$ in
      $Q^{\alpha_0,\beta_0}(x,y)$ at points of $\PK$.}
    
    For each $(\alpha_j,\beta_j) \in S$, we characterize the possible
    $(\alpha_0,\beta_0) \in \PC$ such that $(\alpha_j,\beta_j)$ is a
    root of $Q^{\alpha_0,\beta_0}$ of multiplicity bigger than
    one. Then, we study the $(\alpha_0,\beta_0) \in \PC$ such that
    $(\alpha_0,\beta_0)$ is a root of $Q^{\alpha_0,\beta_0}$ with
    multiplicity bigger than one. Finally, we reduce every case to the
    previous ones.
    
    To study the multiplicities of the roots, we use the fact that
    $(\alpha_0,\beta_0)$ is a double root of a binary form $P$ if and
    only if
    $P(\alpha_0,\beta_0) = \frac{\partial P}{\partial x}(\alpha_0,\beta_0)
    = \frac{\partial P}{\partial y}(\alpha_0,\beta_0) = 0$. Hence, for
    each $(\alpha_0,\beta_0) \in \PC$, we consider
    $\frac{\partial Q^{\alpha_0,\beta_0}}{\partial x}$ and
    $\frac{\partial Q^{\alpha_0,\beta_0}}{\partial y}$, where
    
    \begin{align} \label{eq:derivQLagrangePoly}
      \frac{\partial Q^{\alpha_0,\beta_0}}{\partial x} = &
                                                           -\frac{P_w(\alpha_0,\beta_0)}{P_v(\alpha_0,\beta_0)}
                                                           \frac{1}{M(\alpha_0,\beta_0)}
                                                           \left(
                                                           \frac{\partial M }{\partial x} P_v
                                                           +
                                                           M  \frac{\partial P_v}{\partial x}
                                                           \right)\!(x,y)
                                                           + \\ \nonumber &
                                                                            \sum_{i = 1}^{N_2-N_1} \frac{1}{ \beta_0 \alpha_i  - \alpha_0 \beta_i}
                                                                            \frac{\partial ((\beta_0 x - \alpha_0 y) E_i P_v)}{\partial x}
                                                                            \!(x,y)
                                                                            + \frac{\partial P_w}{\partial x}(x,y)
    \end{align}
    
    Let $O_x^{\alpha_0,\beta_0}(x,y)$ be the product between the last line
    of \Cref{eq:derivQLagrangePoly} and $M(\alpha_0,\beta_0)$, that is
    $$ O_x^{\alpha_0,\beta_0}(x,y) :=
    \sum_{i = 1}^{N_2-N_1} \frac{M(\alpha_0,\beta_0)}{ \beta_0 \alpha_i  - \alpha_0 \beta_i}
    \frac{\partial ((\beta_0 x - \alpha_0 y) E_i P_v)}{\partial x}
    \!(x,y)
    + M(\alpha_0,\beta_0) \frac{\partial P_w}{\partial x}(x,y)
    $$

    In what follows, instead of considering the pair  $(\alpha_0,\beta_0)$ as a
    point in $\PC$, we consider it as a pair of variables. Hence,
    for every $(\alpha_i,\beta_i) \in S$,
    $(\beta_0 \alpha_i - \alpha_0 \beta_i)$ divides
    $M(\alpha_0,\beta_0)$, as polynomials in
    $\closurefield[\alpha_0,\beta_0]$, so
    $O_x^{\alpha_0,\beta_0}(x,y)$ is a polynomial in
    $\closurefield[\alpha_0,\beta_0][x,y]$.
    The derivative of $Q^{\alpha_0,\beta_0}$ with respect to $x$ is a rational 
    function in $\closurefield(\alpha_0,\beta_0)[x,y]$, that we can write as
    $\frac{\partial Q^{\alpha_0,\beta_0}}{\partial x}  = \frac{T^{\alpha_0,\beta_0}(x,y)}{P_v(\alpha_0,\beta_0)
      M(\alpha_0,\beta_0)} $ where
    \begin{align*}
      T^{\alpha_0,\beta_0}(x,y)
      :=                                                              
      -P_w(\alpha_0,\beta_0)
      \left(
      \frac{\partial M }{\partial x} P_v
      +
      M  \frac{\partial P_v}{\partial x}
      \right)\!(x,y)
      + O_x^{\alpha_0,\beta_0}(x,y) P_v(\alpha_0,\beta_0)
      \in \closurefield[\alpha_0,\beta_0][x,y]
    \end{align*}
    \begin{lemma} \label{thm:alphaisqroot}
      For each $(\alpha_i,\beta_i) \in S$, there are at most $N_2+1$
      possible $(\bar\alpha_0,\bar\beta_0) \in \PC$ such that
      $(\bar\alpha_0,\bar\beta_0) \not\in S$, $P_v(\bar\alpha_0,\bar\beta_0) \neq 0$ and
      that $(\alpha_i,\beta_i)$ is a root of multiplicity bigger than $1$
      in $Q^{\bar\alpha_0,\bar\beta_0}$.
    \end{lemma}
    
    \begin{proof}
      If $(\alpha_i,\beta_i)$ is a root of multiplicity bigger than
      $1$ in $Q^{\bar\alpha_0,\bar\beta_0}$, then
      $\frac{\partial Q^{\bar\alpha_0,\bar\beta_0}}{\partial
        x}(\alpha_i,\beta_i) = 0$. Hence, we are looking for the
      $(\bar\alpha_0,\bar\beta_0)$ such that
      $ T^{\bar\alpha_0,\bar\beta_0}(\alpha_i,\beta_i) = 0 $. The
      polynomial $T^{\alpha_0,\beta_0}(\alpha_i,\beta_i)$ belongs to
      $\closurefield[\alpha_0,\beta_0]$, so if it is not identically
      zero, then there are a finite number of points
      $(\bar\alpha_0,\bar\beta_0) \in \PC$ such that
      $ T^{\bar\alpha_0,\bar\beta_0}(\alpha_i,\beta_i) = 0 $.
      Moreover, the degree of the polynomial
      $T^{\alpha_0,\beta_0}(\alpha_i,\beta_i)$ is at most
      $\max(\deg(P_w), \deg(O_x^{\alpha_0,\beta_0}(\alpha_i,\beta_i))
      + \deg(P_v)) = N_2 + 1$. Hence, if the polynomial is not zero,
      this finite number is at most $N_2+1$.
       
      The polynomial
      $T^{\alpha_0,\beta_0}(\alpha_i,\beta_i) \in
      \closurefield[\alpha_0,\beta_0]$ is not zero. Observe that as
      $M$ is square-free, $M(\alpha_i,\beta_i) = 0$ and
      $P_v(\alpha_i,\beta_i) \neq 0$, then
      $\left( \frac{\partial M }{\partial x} P_v + M \frac{\partial
          P_v}{\partial x} \right)\!(\alpha_i,\beta_i) \neq 0$. Hence,
      as $P_w$ and $P_v$ are coprime, if
      $(\bar{\alpha}_0,\bar{\beta}_0) \in \PC$ and
      $P_v(\bar{\alpha}_0,\bar{\beta}_0) = 0$, then
      $T^{\bar{\alpha}_0,\bar{\beta}_0}(\alpha_i,\beta_i) \neq
      0$. That is, $T^{\alpha_0,\beta_0}(\alpha_i,\beta_i)$ does not
      vanish in the roots of $P_v$. 
    \end{proof}
    
    \begin{lemma} \label{thm:choosedDoubleRoot}
      There are at most $2 N_2+1$ possible $(\bar\alpha_0,\bar\beta_0) \in \PC$
      such that $(\bar\alpha_0,\bar\beta_0) \not\in S$, \linebreak
      $P_v(\bar\alpha_0,\bar\beta_0) \neq 0$ and $(\bar\alpha_0,\bar\beta_0)$ is a root of
      multiplicity bigger than $1$ in $Q^{\bar\alpha_0,\bar\beta_0}$.
    \end{lemma}
    
    \begin{proof}
      Following the proof of \Cref{thm:alphaisqroot}, we study
      $T^{\alpha_0,\beta_0}(\alpha_0,\beta_0) \in
      \closurefield[\alpha_0,\beta_0]$.
      \begin{align*}
        T^{\alpha_0,\beta_0}(\alpha_0,\beta_0) = &
                                                   -P_w(\alpha_0,\beta_0)
                                                   \left(
                                                   \frac{\partial M }{\partial x} P_v
                                                   +
                                                   M  \frac{\partial P_v}{\partial x}
                                                   \right)\!(\alpha_0,\beta_0)
                                                   + O_x^{\alpha_0,\beta_0}(\alpha_0,\beta_0) P_v(\alpha_0,\beta_0)
        \\ = &
               \left(
               -P_w \,
               M \,
               \frac{\partial P_v}{\partial x}
               \right)
               \!(\alpha_0,\beta_0)
               +
               \left(
               O_x^{\alpha_0,\beta_0} \,
               -P_w 
               \frac{\partial M }{\partial x}
               \right)\!(\alpha_0,\beta_0)
               P_v(\alpha_0,\beta_0)
      \end{align*}
      Note that $P_w$ and $P_v$ are coprime. Also, $M$ and $P_v$ are
      coprime. Hence, the polynomial
      $T^{\alpha_0,\beta_0}(\alpha_0,\beta_0)$ is not zero because
      $P_v$ does not divide
      $ P_w \, M \, \frac{\partial P_v}{\partial x} $. We conclude the
      proof by noting that the degree of
      $T^{\alpha_0,\beta_0}(\alpha_0,\beta_0)$ is bounded by
      $2 N_2 + 1$.
    \end{proof}
    
    \begin{lemma} \label{thm:sameRootSamePol}
      Let
      $(\bar{\alpha}_0,\bar{\beta}_0),(\hat\alpha_0,\hat\beta_0) \in
      \PP^1(\closurefield)$ such that
      $(\bar{\alpha}_0,\bar{\beta}_0),(\hat\alpha_0,\hat\beta_0) \not\in S$,
      $P_v(\bar{\alpha}_0,\bar{\beta}_0) \neq 0$. Hence,
      $Q^{\bar{\alpha}_0,\bar{\beta}_0}(\hat\alpha_0,\hat\beta_0) = 0$ if and
      only if
      $Q^{\bar{\alpha}_0,\bar{\beta}_0}(x,y) =
      Q^{\hat\alpha_0,\hat\beta_0}(x,y)$.
    \end{lemma}
    
    \begin{proof}
      Assume that
      $Q^{\bar{\alpha}_0,\bar{\beta}_0}(\hat\alpha_0,\hat\beta_0)
      = 0$.  Following \Cref{lem:iterpolMu}, we write
      $Q^{\bar{\alpha}_0,\bar{\beta}_0} = P_{\bar{\mu}} P_v + P_w$ and
      $Q^{\hat\alpha_0,\hat\beta_0} = P_{\hat\mu} P_v + P_w$.
      As $P_v$ and $P_w$ are coprime and
      $Q^{\bar{\alpha}_0,\bar{\beta}_0}(\hat\alpha_0,\hat\beta_0)
      = 0$, then $P_v(\hat\alpha_0,\hat\beta_0) \neq 0$.
      Consider
      $Q^{\bar{\alpha}_0,\bar{\beta}_0} -
      Q^{\hat\alpha_0,\hat\beta_0} = P_v (P_{\bar{\mu}} -
      P_{\hat\mu})$. This polynomial belongs to $\K[x,y]$ and it
      vanishes over $\PC$ at the $N_1 + 1$ roots of $P_v$, at the
      $N_2 - N_1$ points on $S$, and at
      $(\hat\alpha_0,\hat\beta_0) \in \PC$.  Hence,
      $Q^{\bar{\alpha}_0,\bar{\beta}_0} -
      Q^{\hat\alpha_0,\hat\beta_0} = 0$ as it is a binary form in
      $\K[x,y]$ of degree at most $N_2 + 1$ with $N_2 + 2$ different
      roots over $\PC$. Therefore,
      $Q^{\bar{\alpha}_0,\bar{\beta}_0}(x,y) =
      Q^{\hat\alpha_0,\hat\beta_0}(x,y)$.

      To prove the second case, note that, by definition,
      $Q^{\hat\alpha_0,\hat\beta_0}(\hat\alpha_0,\hat\beta_0)
      = 0$. Hence, if we assume that
      $Q^{\bar{\alpha}_0,\bar{\beta}_0}(x,y) =
      Q^{\hat\alpha_0,\hat\beta_0}(x,y)$, then we have 
      $Q^{\bar{\alpha}_0,\bar{\beta}_0}(\hat\alpha_0,\hat\beta_0)
      = 0$.
    \end{proof}
    
    \begin{proof}[Proof of {\Cref{thm:badValuesToInterpolateQ}}]
      We want to bound the number of different points
      $(\bar\alpha_0,\bar\beta_0) \in \PC$ such that
      $Q^{\bar\alpha_0,\bar\beta_0}(x,y)$ is not a square-free binary form
      over $\closurefield[x,y]$. If the binary form
      $Q^{\bar\alpha_0,\bar\beta_0}(x,y)$ is not square-free, then it has a
      root over $\PC$ with multiplicity bigger than one. If such a
      root is $(\alpha_i,\beta_i) \in S$, we can bound the possible
      number of different values for $(\bar\alpha_0,\bar\beta_0) \in \PC$ by
      $(N_2 + 1)$ (\Cref{thm:choosedDoubleRoot}). Hence, if there is a
      $i$ such that $(\alpha_i,\beta_i) \in S$ has multiplicity bigger
      than one as a root of $Q^{\bar\alpha_0,\bar\beta_0}(x,y)$, we can bound
      the possible number of different values for
      $(\bar\alpha_0,\bar\beta_0) \in \PC$ by
      $\# S \cdot (N_2 + 1) = (N_2 - N_1) (N_2 + 1)$.

      If $Q^{\bar\alpha_0,\bar\beta_0}$ is not square-free and the
      multiplicity of every root $(\alpha_i,\beta_i) \in S$ is one,
      then there must be a root
      $(\hat{\alpha}_0,\hat{\beta}_0) \in \PC$ such that
      $(\hat{\alpha}_0,\hat{\beta}_0) \not\in S$ and its multiplicity
      as a root of $Q^{\bar\alpha_0,\bar\beta_0}$ is bigger strictly
      than one.
      By \Cref{thm:sameRootSamePol},
      $Q^{\hat{\alpha}_0,\hat{\beta}_0}(x,y) =
      Q^{\bar\alpha_0,\bar\beta_0}(x,y)$, and so
      $(\hat{\alpha}_0,\hat{\beta}_0) \in \PC$ has multiplicity bigger
      than one as a root of $Q^{\hat{\alpha}_0,\hat{\beta}_0}(x,y)$.
      Hence, $P_v(\hat{\alpha}_0,\hat{\beta}_0) \neq 0$ and, by
      \Cref{thm:choosedDoubleRoot}, we can bound the possible number
      of different values for $(\hat{\alpha}_0,\hat{\beta}_0) \in \PC$
      by $2 N_2 + 1$.
      As $Q^{\hat{\alpha}_0,\hat{\beta}_0}(x,y)$ has $N_1 + 1$ roots
      over $\PC \setminus S$ then, by \Cref{thm:sameRootSamePol},
      there are $N_1 + 1$ different $(\bar\alpha_0,\bar\beta_0) \in \PC$ such
      that
      $Q^{\hat{\alpha}_0,\hat{\beta}_0}(x,y) =
      Q^{\bar\alpha_0,\bar\beta_0}(x,y)$. Hence, for each
      $(\hat{\alpha}_0,\hat{\beta}_0) \in \PC$ such that
      $(\hat{\alpha}_0,\hat{\beta}_0)$ has multiplicity bigger than
      one as a root of $Q^{\hat{\alpha}_0,\hat{\beta}_0}(x,y)$, there
      are $N_1 + 1$ points $(\bar\alpha_0,\bar\beta_0) \in \PC$ such that
      $(\hat{\alpha}_0,\hat{\beta}_0)$ has multiplicity bigger than
      one as a root of $Q^{\bar\alpha_0,\bar\beta_0}(x,y)$. Therefore, the
      number of different values for $(\bar\alpha_0,\bar\beta_0) \in \PC$ such
      that $Q^{\bar\alpha_0,\bar\beta_0}(x,y)$ has a root in $\PC \setminus S$
      with multiplicity bigger than one is bounded by
      $(N_1 + 1) (2 N_2 + 1)$.

      Joining these bounds, we deduce that there are at most
      $(N_2-N_1) (N_2 + 1) + (N_1 + 1) (2 N_2 + 1)$ different
      $(\bar\alpha_0,\bar\beta_0) \in \PC$ such that
      $Q^{\bar\alpha_0,\bar\beta_0}$ is not square-free.  Recalling
      that $N_1 = D - N_2$ and $N_2 \leq D$
      (\Cref{prop:existenceN1andN2}), we can bound
      $(N_2-N_1) (N_2 + 1) + (N_1 + 1) (2 N_2 + 1)$, by \linebreak
      $D^2 + 3 D + 1$.
    \end{proof}

    We can also relate \Cref{thm:badValuesToInterpolateQ} to the
    \emph{Waring locus} of the binary form $f(x,y)$, see
    \cite{carlini_waring_2017}. For example, this proposition improves
    \cite[Proposition~3.8]{carlini_waring_2017} by a factor of two.
    Moreover, it shows that the uniqueness condition from
    \cite[Proposition~3.8]{carlini_waring_2017} misses some
    assumptions to hold \footnote{The authors are not taking into
      account that the lambdas that they use in their proof are not
      unique, and so they give us more degrees of freedom that we can
      use to fix more terms in the decomposition.}.

\subsubsection[Getting the Lambdas]{Complexity of computing $\lambda$} 
  \label{sec:solving_the_lambdas}

  We compute the coefficients $\lambda_j$ of the decomposition by
  solving a linear system involving a transposed Vandermonde matrix
  (Step~\ref{stp:solveLambdas} of \Cref{alg:algorithmDecomp}).  We
  follow \citet{kaltofen1989improved} to write the solution of
  \Cref{eq:gettingLambdas} as the evaluation of a rational function
  over the roots of a univariate polynomial.
  
  \begin{definition}
      Given a polynomial $P(x) := \sum_{i=0}^n a_i x^i$
      and  $0 < k \leq n$, let $$Quo(P(x), x^k) := \sum_{i=k}^{n} a_{i} x^{i-k}.$$
    \end{definition}

    \begin{proposition}[\xspace{\citealp[Sec.~5]{kaltofen1989improved}}]
    \label{prop:solvingTransposeVandermonde}

      If $\alpha_j \neq \alpha_i$, for all $i \neq j$, then there is
      a unique solution to the system of
      \Cref{eq:transposeOfVandermonde}.
      \begin{align}
        {\scriptsize
        \label{eq:transposeOfVandermonde}
        \begin{pmatrix}
          1 & 1 & \cdots & 1\\
          \alpha_1 & \alpha_2 & \cdots & \alpha_r   \\
          \vdots & \vdots & & \vdots  \\ 
          \alpha_1^{r-1} & \alpha_2^{r-1} & \cdots & \alpha_r^{r-1}
         \end{pmatrix}
         \lambda = 
        \begin{pmatrix}
         a_0 \\ a_1 \\ \vdots \\ a_{r-1}
        \end{pmatrix}
        }
      \end{align}
      Moreover, if the
      solution is $\lambda = (\lambda_1, \dots, \lambda_r)^\transpose$ then, $\lambda_j =
      \frac{T}{Q'}(\alpha_j)$
      where $Q'(x)$ is the derivative of
      $Q(x) := \prod_{i=1}^r (x-\alpha_i)$, 
      $R(x) := \sum_{i=1}^{r} a_{r-i} x^{i-1}$ and
      $T(x) := Quo\big(Q(x)\*R(x),x^r\big)$.
    \end{proposition}

    \begin{lemma}
      \label{theo:getLambdas}
      Given a binary form $f(x,y) := \sum_{i = 0}^D \binom{D}{i} a_i
      x^i y^{D-i}$, let $Q$ be a square-free kernel polynomial of degree
      $r$, obtained after \cref{stp:solveLambdas} of
      \Cref{alg:algorithmDecomp}. Assume that $y$ does not divide $Q$.
      Let $\alpha_j$ be the $j$-th roots of $Q(x)$, $Q'(x)$ be the
      derivative of $Q(x)$ and the polynomial \linebreak $T(x) :=
      Quo\big(Q(x)\*R(x),x^r\big)$, with $R(x) := \sum_{i=1}^{r} a_{r-i}
      x^{i-1}$. Then, each $\lambda_j$ from 
      \cref{stp:solveLambdas} in \Cref{alg:gettingVandW}
      can be written as $\lambda_j = \frac{T}{Q'}(\alpha_j)$.
    \end{lemma}

    \begin{proof}
      
      As $y$ does not divide $Q$, we can write it as $Q(x,y) = \prod_i (x -
      \alpha_i y)$, where all the $\alpha_i$ are different. Hence, as
      the $r \times r$ leading principal submatrix of
      \Cref{eq:gettingLambdas} is invertible, we can restrict the
      problem to solve that $r \times r$ leading principal subsystem. This
      system is \Cref{eq:transposeOfVandermonde}. Therefore, the proof
      follows from
      \Cref{prop:solvingTransposeVandermonde}.
    \end{proof}

    \begin{propdef}[Symbolic decomposition]
      \label{def:symbolicDecomp}
      
      Let $Q$ be a square-free kernel polynomial related to a minimal
      decomposition of a binary form $f$ of degree $D$, such that $y$
      does not divide $Q$. In this case, we can write $f$ as
      $$f(x, y) = \sum_{\{\alpha \in \closurefield \mid Q(\alpha)=0\}} \frac{T}{Q'}(\alpha) \* (\alpha x + y)^D .$$
    \end{propdef}

    \begin{remark}\label{thm:ifYdividesQ}
      If the square-free kernel polynomial related to a decomposition
      of rank $r$ is divisible by $y$, we can compute
      $\{\lambda_j\}_{j < r}$ of \Cref{eq:gettingLambdas} as in
      \Cref{theo:getLambdas}, by taking $\frac{Q}{y}$ as the kernel
      polynomial.      
      It is without loss of generality to consider $Q = P_{(u_0, \dots, -1,0)^\transpose}$, because $Q$ is square-free, and so
      $y^2$ can not divide it.
      Hence, $\lambda_r = a_D - \sum_{i =
      0}^{r-2} u_i a_{D-r+i+1}$  \cite[Equation~2.12]{reznick2013length}.
  \end{remark}
    
  To summarize the section, given a binary form $f$ of rank $r$, there
  is a square-free kernel polynomial $Q$ of the degree $r$, such that
  the largest degree of its irreducible factors is bounded by
  $\min(r, D-r+1)$ (\Cref{def:symbolicDecomp}). If $Q(x,y)$ is
  not divisible by $y$, the decomposition is
    $$f(x,y) = \sum\nolimits_{\{\alpha \in \closurefield \mid Q(\alpha)=0\}} \frac{T}{Q'}(\alpha) \* ( \alpha x + y )^D,$$
    
    \noindent where $T$ and $Q'$ are polynomials whose coefficients belong to
    $\field$ and whose degrees are bounded by $r$, defined in
    \Cref{theo:getLambdas}.
    When $y$ divides $Q$, the form is similar.

    \subsubsection{Lower bounds on the algebraic degree}
    \label{sec:lower-bounds-algebr}
    In this section we analyze the tightness of the bound of
    \Cref{cor:dg-irr-Q}. To do so, we construct families of examples
    where the  bound is tight. We present two families of examples. In
    the first one, the decomposition is unique. In the second one, it
    is not.
    
    \begin{proposition}[{\citealp[Theorem~5.2]{heinig1984algebraic}}] \label{thm:weFixPvAndPwAndWeHaveABinaryForm}
      For every pair of relatively prime binary forms, $\bar{P}_v$ and
      $\bar{P}_w$, of degrees $\bar{N}_1 + 1$ and $\bar{N}_2 + 1$,
      $\bar{N}_1 \leq \bar{N}_2$, there is a sequence
      $a = (a_0,\dots,a_{\bar{N}_1 + \bar{N}_2})$ such that
      $N_1^a = \bar{N}_1$, $N_2^a = \bar{N}_2$, and we can consider
      the polynomials $\bar{P}_v$ and $\bar{P}_w$ as the kernel
      polynomials $P_v$ and $P_w$ from \Cref{prop:polInKernel} with
      respect to the family of Hankel matrices $\{H_a^k\}_k$.
    \end{proposition}
    
    \begin{corollary}
      If there is an irreducible binary form of degree $\bar{N}_1+1$
      in $\K[x,y]$, then for every $D > 2 \, \bar{N}_1$, there is a
      binary form $f$ of degree $D$ such that its decomposition is
      unique, its rank $\bar{N}_1+1$, and the degree of the biggest
      irreducible factor of the polynomial $Q$ from
      \Cref{alg:algorithmDecomp} in the decomposition is
      $\min(\bar{N}_1+1, D-\bar{N}_1) = \bar{N}_1+1$. That is, the
      algebraic degree of the minimal decomposition over $\K$ is
      $\bar{N}_1+1$ and the bound of \Cref{cor:dg-irr-Q} is tight.
    \end{corollary}

    \begin{proof}
      Let $\bar{P}_v$ be a irreducible binary form of degree
      $\bar{N}_1 + 1$. Let $\bar{P}_w$ be any binary form of degree
      $\bar{N}_2 + 1 := D - \bar{N}_1 + 1$ relatively prime with
      $\bar{P}_v$. Consider the sequence
      $a = (a_0,\dots,a_{\bar{N}_1 + \bar{N}_2})$ of
      \Cref{thm:weFixPvAndPwAndWeHaveABinaryForm} with respect to
      $\bar{P}_v$ and $\bar{P}_w$, and the binary form
      $f(x,y) := \sum_{i = 0}^{D} {D \choose i} a_i x^{i} y^{D-i}$.
      As $\K$ is of characteristic zero, $\K$ is a perfect field, and
      so, as $\bar{P}_v$ is irreducible, it is square-free. Then, by
      \Cref{lem:rankOfDecomposition}, the rank of the decomposition is
      $N_1^a + 1 = \bar{N}_1 + 1$, and by \Cref{thm:decompIsUnique}
      the decomposition is unique. Following
      \Cref{alg:algorithmDecomp}, the polynomial $Q$ is equal to
      $\bar{P}_v$, which is an irreducible polynomial of degree
      $\bar{N}_1+1$. As $D > 2 \bar{N}_1$, then
      $\min(\bar{N}_1+1, D-\bar{N}_1) = \bar{N}_1+1$ and the bound of
      \Cref{cor:dg-irr-Q} is tight.
    \end{proof}

    \begin{lemma}
      Let $\K = \QQ$ and $p \in \N$ a prime number. Then, there is a
      binary form $f$ of degree $2 (p - 1)$ whose decomposition is not
      unique and the bound of \Cref{cor:dg-irr-Q} is tight.
    \end{lemma}

    \begin{proof}
      Consider the binary form
      $f(x,y) := {2 (p - 1) \choose p-1} x^{p-1} \, y^{p-1}$. Using
      \Cref{alg:gettingVandW}, we obtain $P_v = - y^p$ and
      $P_w = x^p$, $N_1 = N_2 = p - 1$. The polynomial $P_v$ is not
      square-free, so we have to consider a square-free kernel
      polynomial in $\ker(H^{N_2+1})$. Moreover, the rank of the
      decomposition is $N_2 + 1 = p$. Every kernel polynomial in
      $\ker(H^{N_2+1})$ in $\QQ[x,y]$ can be written as
      $\mu_w x^p - \mu_v y^p$ for some $\mu_w, \mu_v \in \QQ$.  We are
      interested in the zeros of these polynomials
      (\cref{stp:solveLambdas} of \Cref{alg:algorithmDecomp}), thus we
      consider coprime $\mu_w, \mu_v \in \ZZ$, as the zeros do not
      change. As we want to consider square-free kernel polynomials,
      neither $\mu_w$ nor $\mu_v$ can be zero, and so
      $(1,0) \in \PP^1(\QQ)$ is not a root of any of these
      polynomials. Hence, we rewrite our polynomial as
      $\frac{1}{\mu_v y^p} (\frac{\mu_w x^p}{\mu_v y^p} - 1)$, and so
      we look for the factorization over $\QQ[z]$ of
      $\frac{\mu_w}{\mu_v} z^p - 1$, where $z = \frac{x}{y}$.  We can
      use the Newton's polygon criterion, e.g.,
      \citet[Chapter~6.3]{cassels1986local}, to show that, if
      $\sqrt[p]{\frac{\mu_w}{\mu_v}} \not\in \QQ$, then
      $\frac{\mu_w}{\mu_v} z^p - 1$ is irreducible over $\QQ[x,y]$ and
      so the degree of its biggest irreducible factor is
      $p > \min(p, 2 \, (p-1)-p+1) = p - 1$ (\Cref{cor:dg-irr-Q}).
      If this is not the case, then
      $\sqrt[p]{|\frac{\mu_w}{\mu_v}|} \in \QQ$, and so we can factor it as
      $$\left(\sqrt[p]{|\frac{\mu_w}{\mu_v}|} \cdot z \right)^p - 1 = 
      \left(\sqrt[p]{|\frac{\mu_w}{\mu_v}|} \cdot z - 1\right) \left(\sum_{i=0}^{p-1}
      \left(\sqrt[p]{|\frac{\mu_w}{\mu_v}|} \cdot z \right)^i\right).$$
      The second factor is irreducible because there is an
      automorphism in $\QQ[x]$ (given by
      $z \mapsto \sqrt[p]{|\frac{\mu_v}{\mu_w}|} z$) that transforms
      it into the $p$-th cyclotomic polynomial, which is irreducible
      as $p$ is prime.  Hence, the biggest irreducible factor of this
      polynomial has degree $p - 1 = \min(p, 2 \, (p-1)-p+1)$ and the
      bound of \Cref{cor:dg-irr-Q} is tight.
    \end{proof}

\subsection{Arithmetic complexity} 
\label{sec:arithmetic_complexity}

  \begin{lemma}[Complexity of \Cref{alg:gettingVandW}]
    \label{lem:ComplexityVandW}

    Given a binary form
    $f = \sum_{i =0}^D {D \choose i} a_i x^i y^{D-i}$ of degree $D$,
    \Cref{alg:gettingVandW} computes $P_v$ and $P_w$ in
    $O(\costMul(D)\*\log(D))$ arithmetic operations.
  \end{lemma}

  \begin{proof}
    The complexity of the algorithm is the complexity of computing the
    rows $(i+1)$, $i$ and $(i-1)$ of the Extended Euclidean algorithm
    between $\sum_{i=0} a_i x^i$ and $x^{D+1}$, where the $i$-th row
    is the first row $i$ such that $\deg(R_i) < \frac{D+1}{2}$
    (\Cref{thm:correctnessPvAndPw}). This can be done using the
    \textit{Half-GCD algorithm}, which computes these rows in
    $O(\costMul(D)\*\log(D))$. For a detailed reference of how this
    algorithm works see \citet[Chapter~6.3]{bostan2017algorithmes} or
    \citet[Chapter~11]{gathen_modern_2013}.
  \end{proof}

  \begin{lemma}[Complexity of computing $Q$]
    \label{lem:complexQ}

    Given the kernel polynomials $P_v$ and $P_w$ from
    \Cref{prop:polInKernel}, we compute a square-free polynomial
    $Q_{\mu} := P_{\mu}\*P_v + P_w$ such that the maximal degree of
    its irreducible factors is bounded by \Cref{cor:dg-irr-Q} in
    $O(\costMul(D)\*\log(D))$.
  \end{lemma}

  \begin{proof}
    To compute the vector $\mu$, we choose randomly $N_2-N_1+1$ linear
    forms in $\K[x,y]$ and we proceed as in \Cref{lem:iterpolMu}.  The
    complexity bound is due to multi-point evaluation and
    interpolation of a univariate polynomial
    \cite[Chapter~10]{gathen_modern_2013}.
  \end{proof}

  \begin{theorem} \label{thm:complexityAlgorithm}
    When the decomposition is unique, that is when the rank is
    $N_1 + 1$, then \Cref{alg:algorithmDecomp} computes
    deterministically a symbolic decomposition (\Cref{def:symbolicDecomp}) of a
    binary form in $O(\costMul(D)\,\log(D))$.
     
    When the decomposition is not unique, that is when the rank is
    $N_2 + 1$ and $N_1 < N_2$, then \Cref{alg:algorithmDecomp} is a
    Monte Carlo algorithm that computes a symbolic decomposition of a
    binary form in \linebreak $O(\costMul(D)\,\log(D))$.
  \end{theorem}

  \begin{proof}
    The first step of the algorithm, in both cases, is to compute the kernel
    polynomials $P_v$ and $P_w$ of \Cref{prop:polInKernel} using
    \Cref{alg:gettingVandW}. By \Cref{lem:ComplexityVandW},
    we compute them deterministically in $O(\costMul(D)\*\log(D))$.

    If $P_v$ is square-free, which means that the decomposition is
    unique, then $Q= P_v$.  Otherwise, in \cref{stp:choosingQ}, we
    need to choose some random values to construct the square-free
    polynomial $Q$ (from the kernel polynomials $P_v$ and $P_w$) in
    $O(\costMul(D)\*\log(D))$ as in \Cref{lem:complexQ}.  This is the
    step that makes the algorithm a Monte Carlo one, as we might fail
    to produce a square-free polynomial $Q$.
    
    In both cases, at \cref{stp:solveLambdas}
    we compute the rational function that describes the solution of
    the system in \Cref{eq:gettingLambdas} in $O(\costMul(D)\*\log(D))$
    (\Cref{theo:getLambdas}).
    At \cref{stp:return} of the algorithm, we return the decomposition.
  \end{proof}
  
  We can bound the probability of error of \Cref{alg:algorithmDecomp}
  using \Cref{thm:badValuesToInterpolateQ}, which bounds the number of
  bad values that lead us to a non square-free polynomial $Q$.
  Moreover, we can introduce a Las Vegas version of
  \Cref{alg:algorithmDecomp} by checking if the values that we choose
  to construct a polynomial $Q$ result indeed a square-free
  polynomial. We can do this check in $O(\costMul(D)\*\log(D))$, by
  computing the GCD between the $Q$ and its derivatives.
  
  \begin{remark}
    If we want to output an approximation of the terms of the minimal
    decomposition, with a relative error of $2^{-\epsilon}$, we can
    use Pan's algorithm \citep{PanOpt02}
    \cite[Theorem~15.1.1]{numMethodsRoots}
    to approximate the roots of
    $Q$. In this case the complexity becomes
    $O\big(D \log^2(D)\big(\log^2(D) + \log(\epsilon)\big)\big)$.
  \end{remark}

  \subsection{Bit complexity}
  \label{sec:bit}
  Let $f \in \ZZ[x, y]$ be a binary form  as in \Cref{eq:bf},
  of degree  $D$ and let $\tau$ be the maximum bitsize of the coefficients $a_i$.
  We study the bit complexity of computing suitable approximations of 
  the $\alpha_j$'s, $\beta_{j}$'s, and $\lambda_j$'s of \Cref{eq:bfDecomp},
  say $\widetilde{\alpha}_j$,  $\widetilde{\beta}_j$ and $\widetilde{\lambda}_j$ respectively,
  that induce an approximate decomposition
  correct up to $\ell$ bits. 
  That is a decomposition such that 
  $\|f - \sum_{j} \widetilde{\lambda}_j (\widetilde{\alpha}_j x + \widetilde{\beta}_j y)^D \|_{\infty} \leq 2^{-\ell}$.
  We need to estimate an upper bound on the number of bits
  that are necessary to perform all the operations of the algorithm.

  % Initially the algorithm performs a linear change of variables,
  % according to the discussion in Sec.~\ref{sub:linearChange}.  The
  % element $t$ of the matrix 
  % $M = \left( \begin{smallmatrix} 1&t\\ 0&1 \end{smallmatrix} \right)$
  % that we use for the linear change of
  % variables has less than $D^3$ forbidden values. Therefore, at least one of
  % the first $D^3$ integers is valid to perform the transformation. 
  % Hence, in the worst case, $t$ has bitsize $O(\log(D)$).
  % After the transformation we obtain a new
  % polynomial, $F$, of the same degree as $f$ and of maximum coefficient
  % bitsize $O(D\lg(D) + \tau) = \sO(D + \tau)$.
  
  % It is wlog to consider $y=1$, because we have already performed the
  % linear change of variables, and the degree of the binary form does
  % not change with this substitution.

  % In the sequel, the algorithm computes the vectors $v$ and $w$ 
  % and, through them, the polynomial $Q$.
  % This costs $\sO(D)$.
  % The degree of $Q$ is $\leq D$ and
  % its maximum coefficient bitsize is $\sO(D^2 + D\tau) = \sigma$
  % as it is the B\'ezout coefficient of an EGCD computation
  % \cite[Sec.~12.3]{gathen_modern_2013}.

  The first step of the algorithm is to compute $P_v$ and $P_w$, via
  the computation of three rows of the Extended GCD of two
  polynomials of degree $D$ and $D+1$ with coefficients of maximal
  sized $\tau$. This can be achieved in $\sOB(D^2 \tau)$ bit
  operations \citep[Corollary~11.14.B]{gathen_modern_2013}, %\citep{wang2003acceleration}
  and the maximal bit size of $P_v$ and
  $P_w$ is $\sO(D \tau)$.  We check if $P_v$ is a square-free
  polynomial in $\sO(D^2 \tau)$, via the computation of the GCD of
  $P_v(x,1)$ and its derivative
  \citep[Corollary~11.14.A]{gathen_modern_2013}, and by checking if $y^2$
  divides it.

  If $P_v$ is square-free polynomial, then $Q = P_v$.  If $P_v$ is not
  square-free, then we can compute $Q$ by assigning values to the
  coefficients of $P_{\mu}$.
  We assume that $y^2$ does not divide $P_w$, if this does not hold,
  we replace $P_w$ by the kernel polynomial $x^{N_2-N_1} P_v + P_w$,
  which is coprime to $P_v$, and so not divisible by $y$, as $P_v$ and
  the original $P_w$ are coprime (\Cref{prop:PvPwCoprimes}). 
  We set all the coefficients of $P_{\mu}$ to zero, except the
  constant term.  Then $Q = \mu_0 y^{N_2-N_1} P_v + P_w$. Now we have
  to choose $\mu_0$ so that $Q$ is square-free. As $y^{N_2-N_1} P_v$
  and $P_w$ are coprime, there are at most $2D + 2$ forbidden values
  for $\mu_0$ such that $Q$ is not square-free
  (\Cref{lem:generalQsqFree}), thus at least one of the first $2D+3$
  integer fits our requirements. We test them all.  Each test
  corresponds to a GCD computation, that costs $\sOB(D^2 \tau)$ and so
  the overall cost is $\sOB(D^3\tau)$.

  Let $\sigma = \sO(D \tau)$ be the maximal bit
  size of $Q$. By \Cref{thm:ifYdividesQ}, we can assume that $y$ does
  not divide $Q$, consider $y = 1$ and treat $Q$ as an univariate
  polynomial.

  Let $\{\alpha_j\}_j$ be the roots of $Q$. 
  We isolate them in $\sOB(D^2 \sigma)$ \citep{PanOpt02}.
  For the (aggregate) separation bound of the roots it holds that
  $-\lg\prod_j\Delta(\alpha_j) = O(D \sigma + D \lg(D))$.
  We approximate all the roots up to accuracy $2^{-\ell_1}$  in
  $\sOB(D^2 \sigma + D \ell_1)$  \citep{pan2017accelerated}. That is, we
  compute absolute approximations of $\alpha_j$,
  say $\widetilde{\alpha}_j$, such that
  $|\alpha_j - \widetilde{\alpha}_j| \leq 2^{-\ell_1}$.
  
  The next step consists in solving the (transposed) Vandermonde
  system, $V(\widetilde{\alpha})^{\transpose} \lambda = a$, where $V(\widetilde{\alpha})$ is
  the Vandermonde matrix we construct with the roots of $Q$, $\lambda$
  is a vector contains the coefficients of decomposition, and $a$ is a
  vector containing the coefficients of $F$, see also \Cref{eq:gettingLambdas}.
  We know the entries of $V(\widetilde{\alpha})$ up to $\ell_1$ bits. Therefore, we
  can compute the elements of the solution vector $\lambda$ with an
  absolute approximation correct up to
  $\ell_2 = \ell_1 - O(D \lg(D) \sigma -\lg\prod_j\Delta(\alpha_j)) = \ell_1 - O(D \lg(D) \sigma)$
  bits \citep[Theorem~29]{pan2017nearly}. 
  That is, we compute $\widetilde{\lambda}_j$'s
  such that  $|\lambda_j - \widetilde{\lambda}_j| \leq 2^{-\ell_2}$.
  At this point we have obtained the approximate decomposition
  $$\widetilde{f}(x,y) := \sum_{j=1}^{r} \widetilde{\lambda}_j (\widetilde{\alpha}_j x +
  y)^D.$$
 % 
  % We apply the inverse transformation, 
  % $M^{-1} = \left( \begin{smallmatrix} 1& -t\\ 0&1 \end{smallmatrix} \right)$,
  % to obtain an approximate decomposition, say $\widetilde{f}$, for $f$,
  % which is 
  % $$ \widetilde{f}(x,y) = \sum\nolimits_{j=1}^{r} \widetilde{\lambda}_j (\widetilde{\alpha}_j x +  (1 - \widetilde{\alpha}_j \, t) y)^D . $$
    
    To estimate the
    accuracy of $\widetilde{f}$ we need to expand the approximate
    decomposition and consider it as a polynomial in $x$. We do not
    actually perform this operation; we only estimate the accuracy as
    if we were. First, we expand each
    $(\widetilde{\alpha}_j x +  y)^D$.
    This results polynomials with coefficients correct up   $\ell_3 =
    \ell_2 - O(D \sigma) = \ell_1 - O(D \lg(D) \sigma) - O(D \sigma) =
    \ell_1 - O(D \lg(D) \sigma)$ bits \citep[Lemma~19]{pan2017nearly}.
    Next, we multiply each such polynomial with
    $\widetilde{\lambda}_i$, and we collect the coefficients for the
    various powers of $x$. Each coefficient is the sum of $r \leq D$
    terms. The last two operations do not affect, asymptotically, the
    precision.
    Therefore, the polynomial 
    $\widetilde{f} =  
    \sum_{j=1}^{r} \widetilde{\lambda}_j (\widetilde{\alpha}_j x +  (1 - \widetilde{\alpha}_j \, t) y)^D$
    that 
    corresponds to the approximate decomposition  has an absolute
    approximation such that  $\| f - \widetilde{f} \| 
    \leq 2^{-\ell_1 + O(D\lg(D) \sigma)}$.
    To achieve an accuracy of $2^{-\ell}$ in the decomposition, such that
    $\|f - \widetilde{f} \| \leq 2^{-\ell}$, we should choose $\ell_1 =
    \ell + O(D\lg(D) \sigma)$. Thus, all the computations should be
    performed with precision of $\ell + O(D\lg(D) \sigma)$ bits.
    The bit complexity of computing the decomposition of $f$ up to
    $\ell$ bits
    is dominated by the solving and refining process and
    it is $\sOB(D \ell + D^2 \sigma)$.
    If we substitute the value for $\sigma$, then we arrive at the complexity bound of 
    $\sOB(D \ell + D^4  + D^3 \tau)$.

    \begin{theorem}
      \label{thm:bf-decomp-bit-compl}
      Let $f \in \ZZ[x, y]$ be a homogeneous polynomial of degree $D$
      and maximum coefficient bitsize $\tau$.  We compute an
      approximate decomposition of accuracy $2^{-\ell}$ in
      $\sOB(D \ell + D^4 + D^3 \tau)$ bit operations.
    \end{theorem}

\section*{Acknowledgments}

The authors thank George Labahn and Vincent Neiger for pointing out
important references to derandomize the computation of $P_v$ and
$P_w$.
The authors thank the anonymous reviewers for their helpful comments.
Mat\'ias Bender thanks Joos Heintz for supervising his Master's
thesis.
The authors are partially supported by ANR JCJC GALOP
(ANR-17-CE40-0009), the PGMO grant ALMA and the PHC GRAPE.

\bibliographystyle{elsarticle-harv}
\bibliography{bibliography}

\begin{thebibliography}{46}
\expandafter\ifx\csname natexlab\endcsname\relax\def\natexlab#1{#1}\fi
\expandafter\ifx\csname url\endcsname\relax
  \def\url#1{\texttt{#1}}\fi
\expandafter\ifx\csname urlprefix\endcsname\relax\def\urlprefix{URL }\fi

\bibitem[{Ansola et~al.(2017)Ansola, Díaz-Cano, and
  Zurro}]{ansola_semialgebraic_2017}
Ansola, M., Díaz-Cano, A., Zurro, M.~A., Jun. 2017. Semialgebraic
  decomposition of real binary forms of a given degree's space.
  arXiv:1706.04207 [math]ArXiv: 1706.04207.
\newline\urlprefix\url{http://arxiv.org/abs/1706.04207}

\bibitem[{Bajaj(1988)}]{bajaj1988algebraic}
Bajaj, C., 1988. The algebraic degree of geometric optimization problems.
  Discrete \& Computational Geometry 3~(1), 177--191.

\bibitem[{Bender et~al.(2016)Bender, Faug{\`e}re, Perret, and
  Tsigaridas}]{bender2016superfast}
Bender, M.~R., Faug{\`e}re, J.-C., Perret, L., Tsigaridas, E., 2016. A
  superfast randomized algorithm to decompose binary forms. In: Proceedings of
  the ACM on International Symposium on Symbolic and Algebraic Computation.
  ACM, pp. 79--86.

\bibitem[{Bernardi et~al.(2017)Bernardi, Daleo, Hauenstein, and
  Mourrain}]{bernardi_tensor_2017}
Bernardi, A., Daleo, N.~S., Hauenstein, J.~D., Mourrain, B., Dec. 2017. Tensor
  decomposition and homotopy continuation. Differential Geometry and its
  Applications 55, 78--105.
\newline\urlprefix\url{http://www.sciencedirect.com/science/article/pii/S0926224517301055}

\bibitem[{Bernardi et~al.(2011)Bernardi, Gimigliano, and
  Ida}]{bernardi2011computing}
Bernardi, A., Gimigliano, A., Ida, M., 2011. Computing symmetric rank for
  symmetric tensors. Journal of Symbolic Computation 46~(1), 34--53.

\bibitem[{Blekherman(2015)}]{blekherman2015typical}
Blekherman, G., 2015. Typical real ranks of binary forms. Foundations of
  Computational Mathematics 15~(3), 793--798.

\bibitem[{Boij et~al.(2011)Boij, Carlini, and Geramita}]{boij2011monomials}
Boij, M., Carlini, E., Geramita, A., 2011. Monomials as sums of powers: the
  real binary case. Proceedings of the American Mathematical Society 139~(9),
  3039--3043.

\bibitem[{Bostan et~al.(2017)Bostan, Chyzak, Giusti, Lebreton, Lecerf, Salvy,
  and Schost}]{bostan2017algorithmes}
Bostan, A., Chyzak, F., Giusti, M., Lebreton, R., Lecerf, G., Salvy, B.,
  Schost, {\'E}., 2017. Algorithmes efficaces en calcul formel. published by
  the Authors.

\bibitem[{Brachat et~al.(2010)Brachat, Comon, Mourrain, and
  Tsigaridas}]{brachat2010symmetric}
Brachat, J., Comon, P., Mourrain, B., Tsigaridas, E., 2010. Symmetric tensor
  decomposition. Linear Algebra and its Applications 433~(11), 1851--1872.

\bibitem[{Cabay and Choi(1986)}]{cabay1986algebraic}
Cabay, S., Choi, D.-K., 1986. {Algebraic computations of scaled Pad{\'e}
  fractions}. SIAM Journal on Computing 15~(1), 243--270.

\bibitem[{Carlini et~al.(2018)Carlini, Catalisano, Chiantini, Geramita, and
  Woo}]{carlini_symmetric_2018}
Carlini, E., Catalisano, M.~V., Chiantini, L., Geramita, A.~V., Woo, Y., 2018.
  Symmetric tensors: rank, {Strassen}'s conjecture and e-computability. Annali
  della Scuola Normale Superiore di Pisa. Classe di scienze 18~(1), 363--390.
\newline\urlprefix\url{https://dialnet.unirioja.es/servlet/articulo?codigo=6389435}

\bibitem[{Carlini et~al.(2017)Carlini, Catalisano, and
  Oneto}]{carlini_waring_2017}
Carlini, E., Catalisano, M.~V., Oneto, A., Jul. 2017. Waring loci and the
  {Strassen} conjecture. Advances in Mathematics 314, 630--662.
\newline\urlprefix\url{http://www.sciencedirect.com/science/article/pii/S0001870816305758}

\bibitem[{Cassels(1986)}]{cassels1986local}
Cassels, J. W.~S., 1986. Local fields. Vol.~3. Cambridge University Press
  Cambridge.

\bibitem[{Comas and Seiguer(2011)}]{comas2011rank}
Comas, G., Seiguer, M., 2011. On the rank of a binary form. Foundations of
  Computational Mathematics 11~(1), 65--78.

\bibitem[{Comon(2014)}]{comon2014tensors}
Comon, P., 2014. Tensors: a brief introduction. IEEE Signal Processing Magazine
  31~(3), 44--53.

\bibitem[{Comon et~al.(2008)Comon, Golub, Lim, and
  Mourrain}]{comon2008symmetric}
Comon, P., Golub, G., Lim, L.-H., Mourrain, B., 2008. Symmetric tensors and
  symmetric tensor rank. SIAM Journal on Matrix Analysis and Applications
  30~(3), 1254--1279.

\bibitem[{Comon and Mourrain(1996)}]{comon1996decomposition}
Comon, P., Mourrain, B., 1996. Decomposition of quantics in sums of powers of
  linear forms. Signal Processing 53~(2), 93--107.

\bibitem[{Draisma et~al.(2016)Draisma, Horobe{\c{t}}, Ottaviani, Sturmfels, and
  Thomas}]{dhost-dgl-16}
Draisma, J., Horobe{\c{t}}, E., Ottaviani, G., Sturmfels, B., Thomas, R.~R.,
  2016. The euclidean distance degree of an algebraic variety. Foundations of
  computational mathematics 16~(1), 99--149.

\bibitem[{D{\"u}r(1989)}]{dur1989}
D{\"u}r, A., Oct 1989. On computing the canonical form for a binary form of odd
  degree. Journal of Symbolic Computation 8~(4), 327--333.

\bibitem[{Garc\'{\i}a-Marco et~al.(2017)Garc\'{\i}a-Marco, Koiran, and
  Pecatte}]{Garcia-Marco:2017:RAS:3087604.3087605}
Garc\'{\i}a-Marco, I., Koiran, P., Pecatte, T., 2017. Reconstruction algorithms
  for sums of affine powers. In: Proceedings of the 2017 ACM on International
  Symposium on Symbolic and Algebraic Computation. ISSAC '17. ACM, New York,
  NY, USA, pp. 317--324.

\bibitem[{Gathen and Gerhard(2013)}]{gathen_modern_2013}
Gathen, J. v.~z., Gerhard, J., 2013. Modern computer algebra. Cambridge
  University Press, Cambridge.

\bibitem[{Giesbrecht et~al.(2003)Giesbrecht, Kaltofen, and
  Lee}]{giesbrecht2003algorithms}
Giesbrecht, M., Kaltofen, E., Lee, W.-s., 2003. Algorithms for computing
  sparsest shifts of polynomials in power, chebyshev, and pochhammer bases.
  Journal of Symbolic Computation 36~(3-4), 401--424.

\bibitem[{Giesbrecht and Roche(2010)}]{giesbrecht2010interpolation}
Giesbrecht, M., Roche, D.~S., 2010. Interpolation of shifted-lacunary
  polynomials. Computational Complexity 19~(3), 333--354.

\bibitem[{Gundelfinger(1887)}]{gundelfinger1887theorie}
Gundelfinger, S., 1887. Zur theorie der bin{\"a}ren formen. Journal f{\"u}r die
  reine und angewandte Mathematik 100, 413--424.

\bibitem[{Hauenstein et~al.(2016)Hauenstein, Oeding, Ottaviani, and
  Sommese}]{hauenstein_homotopy_2016}
Hauenstein, J.~D., Oeding, L., Ottaviani, G., Sommese, A.~J., 2016. Homotopy
  techniques for tensor decomposition and perfect identifiability. Journal für
  die reine und angewandte Mathematik (Crelles Journal) 0~(0).
\newline\urlprefix\url{https://www.degruyter.com/view/j/crll.ahead-of-print/crelle-2016-0067/crelle-2016-0067.xml}

\bibitem[{Heinig and Rost(1984)}]{heinig1984algebraic}
Heinig, G., Rost, K., 1984. Algebraic methods for Toeplitz-like matrices and
  operators. Springer.

\bibitem[{Helmke(1992)}]{helmke1992waring}
Helmke, U., 1992. Waring's problem for binary forms. Journal of pure and
  applied algebra 80~(1), 29--45.

\bibitem[{Iarrobino and Kanev(1999)}]{iarrobino1999power}
Iarrobino, A., Kanev, V., 1999. Power sums, Gorenstein algebras, and
  determinantal loci. Springer.

\bibitem[{Kaltofen and Yagati(1989)}]{kaltofen1989improved}
Kaltofen, E., Yagati, L., 1989. Improved sparse multivariate polynomial
  interpolation algorithms. In: Symbolic and Algebraic Computation. Springer,
  pp. 467--474.

\bibitem[{Kung(1990)}]{kung1990canonical}
Kung, J.~P., 1990. {Canonical forms of binary forms: variations on a theme of
  Sylvester}. Institute for Mathematics and Its Applications 19, 46.

\bibitem[{Kung and Rota(1984)}]{kung1984invariant}
Kung, J.~P., Rota, G.-C., 1984. The invariant theory of binary forms. Bulletin
  of the American Mathematical Society 10~(1), 27--85.

\bibitem[{Landsberg(2012)}]{landsberg2012tensors}
Landsberg, J.~M., 2012. Tensors: geometry and applications. American
  Mathematical Society.

\bibitem[{McNamee and Pan(2013)}]{numMethodsRoots}
McNamee, J.~M., Pan, V.~Y., 2013. {Numerical methods for roots of polynomials
  (II)}. Elsevier.

\bibitem[{Nie et~al.(2010)Nie, Ranestad, and Sturmfels}]{nie2010algebraic}
Nie, J., Ranestad, K., Sturmfels, B., 2010. The algebraic degree of
  semidefinite programming. Mathematical Programming 122~(2), 379--405.

\bibitem[{Oeding and Ottaviani(2013)}]{oeding2013eigenvectors}
Oeding, L., Ottaviani, G., 2013. Eigenvectors of tensors and algorithms for
  waring decomposition. Journal of Symbolic Computation 54, 9--35.

\bibitem[{Pan(2001)}]{pan2001structured}
Pan, V., 2001. Structured matrices and polynomials: unified superfast
  algorithms. Springer.

\bibitem[{Pan(2002)}]{PanOpt02}
Pan, V.~Y., 2002. Univariate polynomials: Nearly optimal algorithms for
  numerical factorization and root-finding. Journal of Symbolic Computation
  33~(5), 701 -- 733.

\bibitem[{Pan and Tsigaridas(2017{\natexlab{a}})}]{pan2017accelerated}
Pan, V.~Y., Tsigaridas, E., 2017{\natexlab{a}}. Accelerated approximation of
  the complex roots and factors of a univariate polynomial. Theoretical
  Computer Science 681, 138--145.

\bibitem[{Pan and Tsigaridas(2017{\natexlab{b}})}]{pan2017nearly}
Pan, V.~Y., Tsigaridas, E.~P., 2017{\natexlab{b}}. Nearly optimal computations
  with structured matrices. Theoretical Computer Science 681, 117--137.

\bibitem[{Reznick(1996)}]{reznick1996homogeneous}
Reznick, B., 1996. Homogeneous polynomial solutions to constant coefficient
  pde's. Advances in Mathematics 117~(2), 179--192.

\bibitem[{Reznick(2013{\natexlab{a}})}]{reznick2013length}
Reznick, B., 2013{\natexlab{a}}. On the length of binary forms. In: Quadratic
  and Higher Degree Forms. Springer, pp. 207--232.

\bibitem[{Reznick(2013{\natexlab{b}})}]{reznick2013some}
Reznick, B., 2013{\natexlab{b}}. Some new canonical forms for polynomials.
  Pacific Journal of Mathematics 266~(1), 185--220.

\bibitem[{Reznick and Tokcan(2017)}]{reznick2017binary}
Reznick, B., Tokcan, N., 2017. Binary forms with three different relative
  ranks. Proceedings of the American Mathematical Society 145~(12), 5169--5177.

\bibitem[{Reznick(1992)}]{reznick1992sums}
Reznick, B.~A., 1992. Sums of even powers of real linear forms. Vol. 463.
  American Mathematical Society.

\bibitem[{Sylvester(1851)}]{sylvester1851remarkablediscovery}
Sylvester, J.~J., 1851. On a remarkable discovery in the theory of canonical
  forms and of hyperdeterminants. The London, Edinburgh, and Dublin
  Philosophical Magazine and Journal of Science 2~(12), 391--410.

\bibitem[{Sylvester(1904)}]{sylvester1851canonicalforms}
Sylvester, J.~J., 1904. An essay on canonical forms, supplement to a sketch of
  a memoir on elimination, transformation and canonical forms. In: The
  collected papers of James Joseph Sylvester. Vol.~1. Cambridge University
  Press, pp. 203--216, the original paper dates back to 1851.

\end{thebibliography}

\end{document}